%% file: asymptoticnormality23.tex
\documentclass[a4paper,11pt]{article}
\usepackage{a4wide}

\usepackage[ansinew]{inputenc}
\usepackage{amssymb}
\usepackage{amsmath}
\usepackage{amsthm}
\usepackage{bbm}
\usepackage{enumerate}
\usepackage{graphicx}

\usepackage[numbers]{natbib}
\numberwithin{equation}{section}

\newtheorem{lemma}{Lemma}
\newtheorem{theorem}{Theorem}
\newtheorem{proposition}{Proposition}
\newtheorem*{corollary}{Corollary}
\theoremstyle{definition}

\newtheorem*{remark}{Remark}

\renewcommand{\qed}{}
\newcommand{\R}{{\mathbbm R}}
\newcommand{\C}{{\mathbbm C}}
\newcommand{\N}{{\mathbbm N}}
\newcommand{\switch}{0}

\begin{document}

\title{Confidence sets in nonparametric calibration of \\exponential L\'evy models}

\author{Jakob S\"ohl\thanks{I thank Markus Rei\ss~for fruitful discussions at all stages of this work and
Mathias Trabs for helpful comments on the manuscript.
I am thankful to the two referees for their comments and suggestions, which led
to improvements of the manuscript.
Most of the project was done at Humboldt--Universit\"at zu Berlin and
the research was supported by the Deutsche Forschungsgemeinschaft through the
SFB 649 ``Economic Risk''.
E-mail address: j.soehl@statslab.cam.ac.uk}\\
University of Cambridge}

\date{\today}

\maketitle

\begin{abstract}
\input{abstract}

\end{abstract}


\input{body}

\bibliography{bib} 

\end{document}

%% file: abstract.tex
Confidence intervals and joint confidence sets are constructed for the
nonparametric calibration of exponential L\'evy models based on prices of
European options. To this end, we show joint asymptotic normality in the
spectral
calibration method for the estimators of the volatility, the drift, the
jump intensity and the L\'evy density at finitely many points.

%% file: body.tex
\section{Introduction}
\label{intro}
The unknown future development of financial markets, as faced by 
participants including investors, traders and companies, can be understood
to consist
of model risk and ``Knightian uncertainty'' \cite{ellsberg,knight}.
The first describes the risk for a given calibrated model and can be evaluated by probabilistic methods, whereas the second incorporates the lack of knowledge on the underlying probability measure and is typically treated by worst case scenarios, for example, by stress testing, which amounts to taking the supremum or infimum over a range of probability measures.

In order to address such questions of robustness it is important to quantify the uncertainty in the underlying probability measure.
By the choice of the model, there is a trade--off between the calibration error and the misspecification of the model.
Both types of uncertainty are unavoidable but within a model the calibration
error is traceable, at least with assumptions on the errors. In general, large
models reduce misspecification, which motivates our choice for a rich
nonparametric model. We then assess the calibration error by constructing
confidence sets.

More precisely, we consider the nonparametric calibration when the risk--neutral price of a stock $(S_t)$ follows an exponential L\'evy model
\begin{align}\label{expLevy}
S_t=Se^{rt+X_t}\quad     \text{with a L\'evy process }X_t  \text{ for }t\ge0.
\end{align}
In this paper we restrict ourselves to L\'evy
processes $X_t$ with finite activity.
Exponential L\'evy models generalize the Black--Scholes model by accounting in
addition to volatility and drift for jumps in the price process.
They are capable of reproducing not only a volatility smile or skew but also the effect that the smile or skew is more pronounced for shorter maturities.
A thorough discussion of this model is given in the monograph by Cont and Tankov~\cite{ContTankov3}.
They introduced in \cite{ContTankov,ContTankov2} a nonparametric calibration method for this model based on prices of European call and put options, in which a least squares approach is penalized by relative entropy. Regularizing by a spectral cut--off, Belomestny and Rei\ss~\cite{calibration} used a different approach to the same estimation problem. Their method achieves the minimax rates of convergence, meaning that their estimators optimize the rate of convergence for the least favorable constellation in a given class of L\'evy processes.
We show asymptotic normality and construct confidence sets and intervals for their estimation procedure. Methods similar to theirs were also applied by Belomestny \cite{belomestny} to estimate the fractional order of regular L\'evy processes of exponential type, by Belomestny and Schoenmakers~\cite{BelomestnySchoenmakers2011} to calibrate a Libor model and by Trabs \cite{trabs:2011} to estimate self--decomposable L\'evy processes.

Confidence sets measure how reliable the estimation is. This is particularly important if the calibrated model is to be used for pricing and hedging.
For a recent review on pricing and hedging in exponential L\'evy models see~\cite{pricing_and_hedging} and the references therein. For the influence of model uncertainty on the pricing see~\cite{cont}.

Nonparametric confidence intervals and sets for L\'evy triplets have not been studied with the notable exception of the work by Figueroa--L{\'o}pez~\cite{figueroa}.
The work is more general in the sense that beyond pointwise confidence intervals
also confidence bands are constructed. On the other hand, the method is based on
direct high--frequency observations so that the statistical problem of
estimating the L\'evy density is easier than in our set--up.
We observe the L\'evy process only indirectly since our method is based on
option prices. This indirect observation scheme does not correspond to direct
observations at high frequency but at low frequency, where the time between
observations is fixed and does not tend to zero. An underlying deconvolution
problem has to be solved and the calibration is a nonlinear inverse problem,
which is mildly ill--posed for volatility zero and severely ill--posed for
positive volatility.

Confidence intervals and sets in nonparametric problems are a subtle issue. Usually some smoothness assumptions on the unknown function are imposed and then the size of a confidence set, or more precisely, the rate at which the confidence set becomes smaller depends on the assumed smoothness. A confidence set adjusting to an unknown smoothness of the estimated function is called adaptive if the confidence set becomes smaller at the same rate as if the smoothness were known. In the nonparametric problem of density estimation Low~\cite{low1997} proved that adaptive confidence intervals do not exist.
Whether in our setting adaptive confidence intervals for the volatility, the
drift and the jump intensity exist is an open question.
We show asymptotic normality for the parametric estimators of the volatility,
the drift and the jump intensity. We also proof asymptotic normality for the
pointwise estimators of the L\'evy density. The joint asymptotic distribution of
these estimators is derived in both the mildly and the severely ill--posed case.
This is used to construct confidence intervals and joint confidence sets.
The asymptotic normality results are based on undersmoothing and on a linearization of the stochastic errors.

The paper is organized as follows. The model and the estimation method are described in Section~\ref{secEstimators}. The main results are formulated in Section~\ref{secNormality}. They are applied to confidence intervals in Section~\ref{secConfidence}. We conclude in Section~\ref{secConclusion}. Some auxiliary results and most proofs are deferred to Section~\ref{secProofs}.

\section{The model and the estimators}\label{secEstimators}

\subsection{The model}

For an underlying $(S_t)$, a strike price $K$ and a maturity $T$, we denote by $\mathcal{C}(K,T)$ and $\mathcal{P}(K,T)$ the prices of European call and put options which are determined by the pricing formula.
We suppose that the risk--neutral price of the stock $(S_t)$ follows the
exponential L\'evy model~\eqref{expLevy} with respect to an equivalent
martingale measure $\mathbb{P}$ and that $X_t$ is a finite activity L\'evy
process.
$S>0$ denotes the present value of the stock and $r\ge0$ the riskless
interest rate.
We fix a maturity $T$ and assume that the observed option prices correspond to
different strike prices $(K_j)$ and are given by the value of the pricing
formula corrupted by noise as motivated by Renault~\cite{renault1997}:
\begin{equation}\label{observationsY}
Y_j=\mathcal{C}(K_j,T)+\eta_j\xi_j, \quad j=1,\dots,n,
\end{equation}
with $\eta_j>0$ and random variables $\xi_j$.
The minimax result in~\cite{calibration} is shown for general errors $(\xi_j)$
which are independent centered random variables with $\mathbb{E}[\xi_j^2]=1$ and
$\sup_j\mathbb{E}[\xi_j^4]<\infty$. The noise levels $(\eta_j)$ can be estimated
nonparametrically, for example, with the method by Fan and
Yao~\cite{FanYao1998}.
As European put and call prices are linked by the put--call parity the observations may alternatively be given by put prices in~\eqref{observationsY}.
We transform the observations to a regression problem on the function
\[\mathcal{O}(x):=\left\{ \begin{array}{ll}
                          S^{-1}\mathcal{C}(x,T), & \quad x\ge0,\\
                          S^{-1}\mathcal{P}(x,T), & \quad x<0,
                          \end{array}
\right.\]
where $x:=\log(K/S)-rT$ denotes the log--forward moneyness.
The regression model may then be written as
\begin{equation}\label{regression}
O_j=\mathcal{O}(x_j)+\delta_j\xi_j,
\end{equation}
where $\delta_j=S^{-1}\eta_j$.

\subsection{The estimation method}

We call the volatility of a L\'evy process~$\sigma$, the drift~$\gamma$ and the
jump intensity~$\lambda$. 
We assume that the jump distribution is absolutely continuous and call its
density~$\nu$.
We denote by $\mu(x):=e^x\nu(x)$ the corresponding exponentially weighted jump
density. The aim is to estimate the L\'evy triplet
$\mathcal{T}=(\sigma^2,\gamma,\mu)$.
In the remainder of this section we present the spectral calibration method of Belomestny and Rei\ss~\cite{calibration}.
The method is based on an option pricing formula by Carr and Madan~\cite{CarrMadan}, which relates the Fourier transform $\mathcal{FO}(u):=\int_{-\infty}^{\infty}\mathcal{O}(x)e^{iux}\mathrm{d}x$ to the characteristic function $\varphi_T(u):=\mathbb{E}[e^{iuX_T}]$.
That is why, we define
\begin{equation}\label{psi}
\begin{split}
\psi(u)&:=\frac{1}{T}\log\left( 1+iu(1+iu)\mathcal{F} \mathcal{O}(u) \right)=\frac1{T}\log(\varphi_T(u-i))\\
&\phantom{:}=-\frac{\sigma^2u^2}{2}+i(\sigma^2+\gamma)u +(\sigma^2/2+\gamma-\lambda)+\mathcal{F}\mu(u),
\end{split}
\end{equation}
where the first equality is given by the above mentioned pricing formula and the
second by the L\'evy--Khintchine representation. This equation links the
observations of $\mathcal{O}$ to the L\'evy triplet that we want to estimate.
Let $\mathcal{O}_n$ be an approximation on the true function $\mathcal{O}$. For
example, $\mathcal{O}_n$ can be obtained by linear interpolation of the
data~\eqref{regression}.
We further define the empirical counterpart of $\psi$ by
\begin{align*}
\psi_{n}(u)&:=\frac{1}{T}\log_{\ge \kappa (u)}\left( 1+iu(1+iu)\mathcal{F} \mathcal{O}_n(u) \right),
\end{align*}
where the trimmed logarithm $\log_{\ge\kappa}:\C\backslash\{0\}\to\C$ is given by
\[\log_{\ge\kappa}(z):=\left\{\begin{array}{ll}\log(z),& |z|\ge\kappa\\\log(\kappa\, z/|z|),&|z|<\kappa\end{array}\right..\]
The logarithms are taken in such a way that $\psi$ and $\psi_{n}$ are continuous
functions with $\psi(0)=\psi_{n}(0)=0$ and $\kappa(u)\in(0,1)$ is specified
in~\cite{calibration}.
Considering \eqref{psi} as a quadratic polynomial in $u$ disturbed by $\mathcal{F}\mu$ motivates the following definition of the estimators for a cut--off value $U>0$:
\begin{align}
\hat{\sigma}^2 &:= \int_{-U}^{U}\mathrm{Re}(\psi_{n}(u))w_{\sigma}^{U}(u)\mathrm{d}u,\label{hat sigma}\\
\hat{\gamma} &:= -\hat{\sigma}^2 +\int_{-U}^{U}\mathrm{Im}(\psi_{n}(u))w_{\gamma}^{U}(u)\mathrm{d}u,\label{hat gamma}\\
\hat{\lambda} &:= \frac{\hat{\sigma}^2}{2}+\hat\gamma -\int_{-U}^{U}\mathrm{Re}(\psi_{n}(u))w_{\lambda}^{U}(u)\mathrm{d}u,\label{hat lambda}
\end{align}
where the weight functions $w_{\sigma}^{U}$, $w_{\gamma}^{U}$ and $w_{\lambda}^{U}$ satisfy
\begin{equation}\label{weight}
\begin{split}
\int_{-U}^{U} & \frac{-u^2}{2} w_{\sigma}^{U}(u)\mathrm{d}u=1,\quad
\int_{-U}^{U} u w_{\gamma}^{U}(u)\mathrm{d}u=1,\quad
\int_{-U}^{U} w_{\lambda}^{U}(u)\mathrm{d}u=1,\\
\int_{-U}^{U} & w_{\sigma}^{U}(u)\mathrm{d}u=0,\quad
\int_{-U}^{U} u^2 w_{\lambda}^{U}(u)\mathrm{d}u=0.
\end{split}
\end{equation}
The estimator for $\mu$ is defined by a smoothed inverse Fourier transform of the remainder
\begin{equation}\label{hat mu}
\hat\mu(x):=\mathcal{F}^{-1}\left[\left(\psi_{n}(u) +\frac{\hat\sigma^2}{2}(u-i)^2-i\hat\gamma(u-i) +\hat\lambda\right)w_\mu^{U}(u)\right](x).
\end{equation}
The choice of the weight functions is discussed in~\cite{soehltrabs}, where also possible weight functions are given.
The weight functions for all $U>0$ can be obtained from $w_\sigma^1$, $w_\gamma^1$, $w_\lambda^1$ and $w_\mu^1$ by rescaling:
\begin{align*}
&w_{\sigma}^{U}(u)=U^{-3}w_{\sigma}^{1}(u/U), \quad w_{\gamma}^{U}(u)=U^{-2}w_{\gamma}^{1}(u/U), \\ &w_{\lambda}^{U}(u)=U^{-1}w_{\lambda}^{1}(u/U), \quad w_\mu^U(u)=w_\mu^1(u/U).
\end{align*}
Since  $\psi_{n}(-u)=\overline{\psi_{n}(u)}$ only the symmetric part of
$w_\sigma^1$, $w_\lambda^1$ and the antisymmetric part of $w_\gamma^1$ matter.
The antisymmetric part of $w_\mu^1$ contributes a purely imaginary part to~$\hat
\mu(x)$. Without loss of generality we will always assume $w_\sigma^1$,
$w_\lambda^1$, $w_\mu^1$ to be symmetric and $w_\gamma^1$ to be antisymmetric.
We further assume that the support of $w_\sigma^1$, $w_\gamma^1$, $w_\lambda^1$ and $w_\mu^1$ is contained in $[-1,1]$.
We define the estimation error $\Delta\hat{\sigma}^2:=\hat{\sigma}^2 -\sigma^2$ and likewise for the other estimators. We will also use the notation $\Delta\psi_{n}:=\psi_{n}-\psi$.
The estimation error $\Delta\hat{\sigma}^2$ can be decomposed as
\begin{equation}\label{error_sigma}
\begin{split}
\Delta\hat{\sigma}^2: &=\frac{2}{U^{2}}\int_{0}^{1}\mathrm{Re}(\mathcal{F}\mu(Uu))w_{\sigma}^{1}(u)\mathrm{d}u 
+ \frac{2}{U^{2}}\int_{0}^{1} \mathrm{Re}(\Delta\psi_{n}(Uu) ) w_{\sigma}^{1}(u)\mathrm{d}u.
\end{split}
\end{equation}
The first term is the approximation error and decreases in the cut--off value $U$ due to the decay of $\mathcal{F}\mu$. The second is the stochastic error and increases in $U$ by the growth of  $\Delta\psi_{n}$. For growing sample size $n$ the term $\Delta\psi_{n}$ becomes smaller so that the stochastic error decays even if we let $U\to\infty$ as $n\to\infty$. For $\sigma=0$ the term $\Delta\psi_{n}(u)$ grows polynomial in $u$ so that we can let $U$ tend polynomially to infinity, whereas for $\sigma>0$ it grows exponentially in $u$ and we can let $U$ tend only logarithmically to infinity. This is the reason for the polynomial and logarithmic convergence rates for $\sigma=0$ and for $\sigma>0$, respectively.
For fixed sample size the cut--off value $U$ is the crucial tuning parameter in this method and allows a trade--off between the error terms. The influence of the cut--off value $U$ is analogous to the influence of the bandwidth $h$ on kernel estimators, more precisely $U^{-1}$ corresponds to $h$.
The other estimation errors allow similar decompositions as $\Delta\hat{\sigma}^2$ in~\eqref{error_sigma}.

We shall analyze the asymptotic properties of the stochastic errors in depth. To bound the approximation errors some smoothness assumption is necessary.
We assume that the L\'evy triplet belongs to a smoothness class $\mathcal{G}_s(R,\sigma_{\max})$ with $s\in\N$ and $R,\sigma_{\max}>0$ specified in~ \cite[Definition~4.1]{calibration}.
%
%
The assumption $\mathcal{T}\in\mathcal{G}_s(R,\sigma_{\max})$ includes a smoothness assumption of order $s$ on $\mu$ leading to a decay of $\mathcal{F}\mu$.
To profit from this decay when bounding the approximation error, we assume the weight functions to be of order $s$, this means
\begin{align}\label{order s}
\mathcal{F}(w^1_\sigma(u)/u^s), \mathcal{F}(w^1_\gamma(u)/u^s), \mathcal{F}(w^1_\gamma(u)/u^s),  \mathcal{F}\left((1-w_\mu^{1}(u))/u^s\right)\in L^1(\R).
\end{align}

\subsection{Discussion of the model}

In this paper we restrict to the nonparametric calibration of finite activity
L\'evy processes. The nonlinear penalized least--squares method by Cont and
Tankov \cite{ContTankov} and the spectral calibration method by Belomestny and
Rei\ss\ \cite{calibration} are mainly considered for finite activity L\'evy
processes. Trabs \cite{trabs:2011} extended the spectral calibration method to
self--decomposable L\'evy processes, which have infinite activity
and Blumenthal--Getoor index zero. Extensions to higher
Blumenthal--Getoor indexes are of interest but it might be difficult to
distinguish statistically between volatility and small jumps.
We define the measure
$\nu_\sigma(\mathrm{d}x):=\sigma^2\delta_0(\mathrm{d}x)+x^2\nu(x)\mathrm{d}x$,
where $\delta_0$ denotes the Dirac measure at zero.
Its structure in a neighborhood of zero is very natural since it is most useful
in characterizing weak convergence of the distribution of the L\'evy process in
view of Theorem~VII.2.9 and Remark~VII.2.10 in Jacod and
Shiryaev~\cite{JacodShiryaev2003}. The measure $\nu_\sigma$ determines the
variance of a L\'evy process and is relevant for calculating the~$\Delta$ in
quadratic hedging as noted in~\cite{NeumannReiss2009}. Volatility and small
jumps both contribute to the mass assigned by $\nu_\sigma$ to a neighborhood of
zero. 
One possibility to separate the jumps and the volatility is to assume
finite jump activity. While other assumptions are possible some restriction is
necessary here.
Indeed, Neumann and Rei\ss~\cite{NeumannReiss2009} point out in their Remark~3.2
that without further restrictions the volatility cannot be estimated
consistently.
In Section~2.3 of \cite{soehl} the spectral calibration method designed for
finite intensity
processes is applied to some infinite activity L\'evy processes,
namely to symmetric stable L\'evy processes.
This suggests that in the misspecified case of infinite activity $\hat\sigma^2$
has to be interpreted as the joint quantity of $\sigma^2$ and the small jumps
or, more precisely, as the mass assigned by $\nu_\sigma$ to a neighborhood of
zero with size proportional to $U^{-1}$. In this case $\hat\nu$ should be
consider only outside this neighborhood as an estimator for $\nu$.

\section{Asymptotic normality}\label{secNormality}

\subsection{The main results}

The aim of this section is to establish asymptotic normality results for the estimators. We would like to state that the appropriately scaled errors of the estimators converge to normal random variables.
The starting point of our error analysis is the decomposition
\eqref{error_sigma} into the approximation error and the stochastic error.
The approximation error is deterministic and only the stochastic error can be expected to converge with appropriate scaling to a normal random variable.
It is common practice to resolve this problem by undersmoothing, which means that the tuning parameter is chosen such that the approximation error becomes asymptotically negligible.

To simplify the asymptotic analysis of the stochastic errors, we do not work
with the regression model \eqref{regression} but with the Gaussian white noise
model. This is an idealized observation scheme, where the terms are easier to
analyze.
At the same time asymptotic results may be transferred to the regression model. The Gaussian white noise model is given by
\begin{align}\label{whitenoise}
\mathrm{d}Z_{n}(x)=\mathcal{O}(x)\mathrm{d}x+\epsilon_n \;
\rho(x)\mathrm{d}W(x),
\end{align}
where $W$ is a two--sided Brownian motion, $\rho \in L^2(\R)$ and
$\epsilon_n>0$.
In the case of equidistant design the precise connection to the regression
model~\eqref{regression} is given by $\rho(x_j)=\delta_j$ and
$\epsilon_n=\tfrac{n+1}{n-1}(x_n-x_1)n^{-1/2}$, where $x_1$ and $x_n$ are the
minimal and maximal design points and where we assume that the range of
observations $(x_n-x_1)$ grows slower than $n^{1/2}$ such that
$\epsilon_n\to0$ as $n\to\infty$.
General designs $x_j=F^{-1}(j/(n+1))$ for a c.d.f.
$F:\R\to[0,1]$ with p.d.f. $f$ can be treated by the Gaussian white noise model
$\mathcal{O}(x)\mathrm{d}x+n^{-1/2}\delta(x) f(x)^{-1/2}\mathrm{d}W(x)$,
where $\delta(x_j)=\delta_j$.
Transferring asymptotic results from the Gaussian white noise model to the
regression model is formally justified by the concept of asymptotic equivalence,
which applies in particular to lower bounds and confidence statements.
Brown and Low~\cite{BrownLow} show that the regression model~\eqref{regression}
with Gaussian errors is asymptotically equivalent to the Gaussian white noise
model~\eqref{whitenoise}.
For non--Gaussian errors we refer to Grama and Nussbaum~\cite{GramaNussbaum}.
Their main assumption on the errors is slightly more than Hellinger
differentiability, which is a smoothness assumption on the distributions of the
errors.

The simplified approach of using the Gaussian white noise model to construct
confidence sets is well suited to derive asymptotic
normality and to determine the quantitative expression of the asymptotic
variance. Nevertheless, it is an idealized model and the ultimate interest is in
the regression model.
Obtaining results directly in the regression model would probably lead to less
assumptions than combining the Gaussian white noise model with an asymptotic
equivalence result.
Details of the application of the asymptotic equivalence result can be found
at the beginning of Section~\ref{secProofs}.

The stochastic errors involve the term $\Delta\psi_{n}(Uu)=\psi_n(Uu)-\psi(Uu)$,
which is a
difference of two logarithms.
In the definition of $\psi_n$, we take
$\mathcal{FO}_n=\mathcal{F}(\mathrm{d}Z_n)$ and thus define the empirical
version of $\mathcal{FO}$ directly without constructing in an intermediate step
an empirical version $\mathcal{O}_n$ of $\mathcal{O}$.
For $z,z'\in\C\backslash \{0\}$ and $\kappa>0$ it holds $\log_{\ge
\kappa}(z)-\log(z')=\log_{\ge \kappa/|z'|}\left(z/z'\right)$.
That yields
\begin{align*}
\Delta\psi_{n}(Uu)
&=\frac{1}{T}\log_{\ge \kappa^{U}(u)}\Big(1+\frac{\epsilon_n \:
iUu(1+iUu)}{1+iUu(1+iUu)\mathcal{FO}(Uu)}\int_{-\infty}^{\infty}e^{iUux}\rho (x)
\mathrm{d}W(x)\Big),
\end{align*}
where $\kappa^{U}(u):=\kappa(Uu)/|1+iUu(1+iUu)\mathcal{FO}(Uu)|\le 1/2$, see~\cite[(6.3)]{calibration}.
We define a linearization $\mathcal{L}_{n,U}$ of the logarithm and the remainder term $\mathcal{R}_{n,U}$ by
\begin{align}
\mathcal{L}_{n,U}(u)&:=\frac{\epsilon_n \:
iUu(1+iUu)}{T(1+iUu(1+iUu)\mathcal{FO}(Uu))}\int_{-\infty}^{\infty}e^{iUux}\rho
(x) \mathrm{d}W(x),\label{eqLinearization}\\
\mathcal{R}_{n,U}(u)&:=\Delta\psi_{n}(Uu)-\mathcal{L}_{n,U}(u).\label{eqRemainder}
\end{align}

To ensure continuity of the Gaussian process $X(u)=\int_{-\infty}^\infty
e^{iux}\rho(x)\mathrm{d}W(x)$
we assume that there is a $p>1$ such that
$\int_{-\infty}^{\infty}(1+|x|)^{p}\rho(x)^2\mathrm{d}x<\infty$.
In \cite{polar} it is shown that on this assumption $X$ satisfies the Kolmogorov--Chentsov criterion~\cite[p. 57]{kallenberg} and thus has a continuous version. In the sequel we are always working with this version.

The remainder term $\mathcal{R}_{n,U}$ in~\eqref{eqRemainder} is small when the argument of the logarithm is close to one, that is when $\mathcal{L}_{n,U}$ is small.
Since we are integrating over the unit interval in \eqref{error_sigma} we want $\mathcal{L}_{n,U}$ to be uniformly small.
We shall use the notation $A(x)\lesssim B(x)$ as $x\to\infty$ synonymously with the Landau notation $A(x)=O(B(x))$ as $x\to\infty$, meaning that there exist $M>0$ and $x_0\in\R$ such that $A(x)\le M B(x)$ for all $x\ge x_0$.

\begin{proposition}{\label{sup L}}
For all $q\ge1$ holds
\[\mathbb{E}\left[\sup_{u\in[-1,1]}|\mathcal{L}_{n,U}(u)|^q\right]^{1/q} \lesssim \left\{ \begin{array}{ll}
\epsilon_n U^2 \sqrt{\log(U)},        & \quad \text{for } \sigma=0,\\
\epsilon_n U^2 \exp(T\sigma^2U^2/2),  & \quad \text{for } \sigma>0,
\end{array}
\right.\]
as $U\to\infty$.
\end{proposition}

This proposition is proved in Section~\ref{secProofs} by metric entropy arguments.
In the following theorems we control the supremum of $\mathcal{L}_{n,U}$ and thus the remainder term $\mathcal{R}_{n,U}$ by the conditions
$\epsilon_n U_n^2\sqrt{\log(U_n)} \to 0$ and $\epsilon_n U_n^2\exp(T\sigma^2U_n^2/2) \to 0$ for $\sigma=0$ and for $\sigma>0$, respectively.
Then the asymptotic distribution of the stochastic errors $\int_0^1\Delta\psi_{n}(Uu)w(u)\mathrm{d}u$ is governed by the linearized stochastic errors $\int_0^1\mathcal{L}_{n,U}(u)w(u)\mathrm{d}u$ and the remainder term $\int_0^1\mathcal{R}_{n,U}(u)w(u)\mathrm{d}u$ is asymptotically negligible. In the case $\sigma=0$ the stronger condition $\epsilon_n U_n^{5/2}\to0$ is assumed, which is needed for the stochastic errors to converge to zero.

For the approximation error to be asymptotically negligible we need to undersmooth by choosing the cut--off value $U_n$ large enough such that $\epsilon_n U_n^{(2s+5)/2}\to\infty$ in the case of volatility zero and by $\epsilon_n U_n^{s+1}\exp(T\sigma^2U_n^2/2)\to\infty$ in the case of positive volatility, where $s$ is the smoothness assumed on the exponentially weighted jump measure $\mu$.

In the results on asymptotic normality we will also include the estimator
$\hat\mu(0)$ of the jump density at zero. This only makes sense by our
smoothness assumption on $\mu$ since there is no way of detecting jumps of
height zero. 
The asymptotic distribution of $\hat\mu(0)$ is not determined by the weight
function $w^1_\mu$ but by the effective weight function
\[w_0(u):=w_\mu^1(u)+w_\sigma^1(u)\int_{-1}^1
v^2w_\mu^1(v)\mathrm{d}v/2-w_\lambda^1(u)\int_{-1}^1w_\mu^1(v)\mathrm{d}v.\]

The first theorem states the joint asymptotic normality result for the mildly ill--posed case of volatility zero.

\begin{theorem}\label{asymptotic distribution =0}
Let $\sigma=0$. Let $\rho$ be continuous at
$T\gamma,x_1+T\gamma,\dots,x_m+T\gamma$ and let $\mathcal{F}\rho^2\in
L^1(\R)$. For $j=1,\dots,m$ let $x_j\in\R\backslash\{0\}$ be distinct and let
$V_0,W_0,W_{x_1}\dots,W_{x_m}$ be independent Brownian motions.
If $\epsilon_n U_n^{5/2}\to0$ and $\epsilon_n U_n^{(2s+5)/2}\to\infty$ as $n\to\infty$, then
\begin{align*}
\frac1{\epsilon_n}\left(\begin{array}{rc}
U_n^{+1/2}&\Delta\hat\sigma^2 \\
U_n^{-1/2}&\Delta\hat\gamma^{\phantom{2}} \\
U_n^{-3/2}&\Delta\hat\lambda^{\phantom{2}}\\
U_n^{-5/2}&\Delta\hat\mu(0)\\
U_n^{-5/2}&\Delta\hat\mu(x_1)\\
&\vdots\\
U_n^{-5/2}&\Delta\hat\mu(x_m)
\end{array}\right)
\xrightarrow{d}
\left(\begin{array}{rl}
d(0)&\int_0^1u^2w_\sigma^1(u)\mathrm{d}W_0(u)\\
d(0)&\int_0^1u^2w_\gamma^1(u)\mathrm{d}V_0(u)\\
d(0)&\int_0^1u^2w_\lambda^1(u)\mathrm{d}W_0(u)\\
d(0)&\int_0^1u^2w_0(u)\mathrm{d}W_{0}(u)/(2\pi)\\
d(x_1)&\int_0^1u^2w_\mu^1(u)\mathrm{d}W_{x_1}(u)/(2\pi)\\
&\vdots\\
d(x_n)&\int_0^1u^2w_\mu^1(u)\mathrm{d}W_{x_m}(u)/(2\pi)
\end{array}\right)
\end{align*}
as $n\to\infty$,
where $d(x):=2\sqrt{\pi}\rho(x+T\gamma)\exp(T(\lambda-\gamma))/T$.
\end{theorem}
\begin{remark}
  The theorem is formulated in terms of the exponentially weighted jump density $\mu(x)=e^x\nu(x)$. By the continuous mapping theorem results on $\mu$ can be reformulated in terms of~$\nu$ by multiplying with $e^{-x_j}$ in the respective lines.
\end{remark}
\begin{proof}
We write $\Delta\hat\gamma$, $\Delta\hat\lambda$ and $\Delta\hat\mu(x)$ similarly as in \eqref{error_sigma} for $\Delta\hat\sigma^2$:
\begin{align}
\Delta\hat{\gamma} &=-\Delta\hat{\sigma}^2 +\frac{2}{U}\int_{0}^{1}\mathrm{Im}(\mathcal{F}\mu(Uu))w_{\gamma}^{1}(u)\mathrm{d}u 
+ \frac{2}{U}\int_{0}^{1} \mathrm{Im}(\Delta\psi_{n}(Uu) ) w_{\gamma}^{1}(u)\mathrm{d}u,\label{error_gamma}\displaybreak[0]\\
\Delta\hat{\lambda} &=\frac{\Delta\hat{\sigma}^2}{2}+\Delta\hat\gamma -2\int_{0}^{1}\mathrm{Re}(\mathcal{F}\mu(Uu))w_{\lambda}^{1}(u)\mathrm{d}u 
- 2\int_{0}^{1} \mathrm{Re}(\Delta\psi_{n}(Uu) ) w_{\lambda}^{1}(u)\mathrm{d}u,\label{error_lambda}\displaybreak[0]\\
\Delta\hat\mu(x)
&=U\mathcal{F}^{-1}\left[\Delta\psi_{n}(Uu)w_\mu^{1}(u)\right](Ux)\nonumber\\
&\phantom{=}
+\frac{\Delta\hat\sigma^2}{2} U\mathcal{F}^{-1}\left[(Uu-i)^2w_\mu^{1}(u)\right](Ux)
-i\Delta\hat\gamma U\mathcal{F}^{-1}\left[(Uu-i)w_\mu^{1}(u)\right](Ux) \nonumber\\
&\phantom{=}+\Delta\hat\lambda U\mathcal{F}^{-1}\left[w_\mu^{1}(u)\right](Ux) -U\mathcal{F}^{-1}\left[(1-w_\mu^{1}(u))\mathcal{F}\mu(Uu)\right](Ux). \label{error_mu}
\end{align}
In \eqref{error_gamma} we can substitute $\Delta\hat\sigma^2$ using \eqref{error_sigma} and obtain two error terms involving $\mathcal{F}\mu$ and two error terms involving $\Delta\psi_{n}$. By similar substitutions in \eqref{error_lambda} and \eqref{error_mu} we see that all error terms either involve $\mathcal{F}\mu$ or $\Delta\psi_{n}$, which we will call approximation errors and stochastic errors, respectively.

The undersmoothing $\epsilon_n U_n^{(2s+5)/2}\to\infty$ is equivalent to $U_n^{-(s+3)}=o(\epsilon_n U_n^{-1/2})$. The approximation error of $\hat\sigma^2$ decays by~\eqref{bias_sigma} below as $U_n^{-(s+3)}$ and thus is asymptotically negligible. The approximation errors
$2U^{-1}\int_{0}^{1}\mathrm{Im}(\mathcal{F}\mu(Uu))w_{\gamma}^{1}(u)\mathrm{d}u$ of $\hat\gamma$,
$2\int_{0}^{1}\mathrm{Re}(\mathcal{F}\mu(Uu))w_{\lambda}^{1}(u)\mathrm{d}u$ of $\hat\lambda$
and
$U\mathcal{F}^{-1}\left[(1-w_\mu^{1}(u))\mathcal{F}\mu(Uu)\right](Ux)$ of $\hat\mu(x)$ can be bounded similarly as done in \eqref{bias_gamma}, \eqref{bias_lambda} and \eqref{bias_mu} below and are asymptotically negligible, too.
Since $\hat\sigma^2$ converges with a faster rate than $\hat\gamma$ and $\hat\gamma$ converges with a faster rate than $\hat\lambda$, the error $\Delta\hat\sigma^2$ vanishes asymptotically in \eqref{error_gamma} and in \eqref{error_lambda} as well as $\Delta\hat\gamma$ is asymptotically negligible in \eqref{error_lambda}.
For $x\neq0$ we can apply the Riemann--Lebesgue lemma to the second, the third and the fourth error term in \eqref{error_mu} and we see that they are of order $o_{\mathbb{P}}(\epsilon_n U_n^{5/2})$. For $x=0$ due to the symmetry of $w_\mu^1$ the third term vanishes asymptotically but the second and the fourth term do not. The error terms of $\hat\mu(x)$ we have to consider are in the case $x\neq0$
\begin{align*}
U\mathcal{F}^{-1}\left[\Delta\psi_{n}(Uu)w_\mu^{1}(u)\right](Ux)
=\frac{U}{2\pi}2\int_0^1w_\mu^{1}(u)\mathrm{Re}\left(\Delta\psi_{n}(Uu)e^{-iUux}\right)\mathrm{d}u
\end{align*}
and in the case $x=0$
\begin{align*}
  &\phantom{\;=\;}\mathcal{F}^{-1}\left[\Delta\psi_{n}(Uu)
  )w_\mu^{1}(u)\right](0)
  + \int_{0}^{1} \mathrm{Re}(\Delta\psi_{n}(Uu) )
  w_{\sigma}^{1}(u)\mathrm{d}u\, \mathcal{F}^{-1}\left[u^2w_\sigma^{1}(u)\right](0)\\
  &\phantom{\;=\;}- 2\int_{0}^{1} \mathrm{Re}(\Delta\psi_{n}(Uu) )
  w_{\lambda}^{1}(u)\mathrm{d}u\, \mathcal{F}^{-1}\left[w_\mu^{1}(u)\right](0)\\
  &=\frac{1}{2\pi} 2\int_{0}^{1} \mathrm{Re}(\Delta\psi_{n}(Uu) )
  w_0(u)\mathrm{d}u.
\end{align*}
By assumption~\eqref{order s} on the order of the weight functions, $w_\sigma^1$, $w_\gamma^1$, $w_\lambda^1$ and $w_\mu^1$ are continuous and bounded, especially they are Riemann--integrable and in $L^\infty([-1,1])$.
As the main technical step, Lemma~\ref{linearized_0} shows the convergence of the linearized stochastic errors. The remainder terms are asymptotically negligible by Lemma~\ref{linearization_0}.\qed
\end{proof}

Next we consider the case $\sigma>0$.
Let $\rho$ be in $L^\infty(\R)$ and $\|\rho\|_{L^2(\R)}>0$. We set
\begin{align}
  d&:=\sqrt{2}\|\rho\|_{L^2(\R)} 
\exp(T(\lambda-\gamma-\sigma^2/2))T^{-2}\sigma^{-2}\label{d}
\end{align}
and define by
$W_{n,U}+iV_{n,U}:=2d^{-1}\int_0^1\mathcal{L}_{n,U}(u)\mathrm{d}u$ the real--valued random variables $W_{n,U}$ and $V_{n,U}$.
By Lemma~\ref{linearized_>0} below
\begin{align}
\frac1{\epsilon_n\exp(T\sigma^2U^2/2)}\left(\begin{array}{c}W_{n,U}\\V_{n,U}\end{array}\right) &\xrightarrow{d}
\left(\begin{array}{c}W\\V\end{array}\right) \label{W and V}
\end{align}
as $U\to\infty$,
where $W$ and $V$ are independent standard normal random variables.

The following theorem treats the stochastic errors in the case of positive volatility.
Since the theorem contains no statement on the approximation errors,
the condition~\eqref{order s} on the order of the weight functions may be
omitted.
\begin{theorem}\label{stochastic errors}
Let $\sigma>0$ and $\rho\in L^\infty(\R)$.
Assume for the cut--off value $U_n\to\infty$ and $\epsilon_n U_n^2\exp(T\sigma^2U_n^2/2) \to 0$ as $n\to \infty$.
Let $w_\sigma^1,w_\gamma^1,w_\lambda^1,w_\mu^1:[0,1]\to\R$ be Riemann--integrable, in $L^\infty([0,1])$ and continuous at one.
Then for $x\in\R$
\begin{align*}
  &\frac1{\epsilon_n \exp(T\sigma^2U_n^2/2)}
  \left(\begin{array}{r}
  2\int_0^1\mathrm{Re}(\Delta\psi_{n}(U_n u))w_\sigma^1(u)\mathrm{d}u  -d\,w_\sigma^1(1)W_{n,U_n}\phantom{(x)/(2\pi)}\\
  2\int_0^1\mathrm{Im}(\Delta\psi_{n}(U_n u))w_\gamma^1(u)\mathrm{d}u  -d\,w_\gamma^1(1)\,\,V_{n,U_n }\phantom{(x)/(2\pi)}\\
  2\int_0^1\mathrm{Re}(\Delta\psi_{n}(U_n u))w_\lambda^1(u)\mathrm{d}u  -d\,w_\lambda^1(1)W_{n,U_n }\phantom{(x)/(2\pi)}\\
  \mathcal{F}^{-1}\left[\Delta\psi_{n}(U_n u)
  w_\mu^{1}(u)\right](U_n x) -d\,w_\mu^1(1)\:Z_{n,U_n }(x)/(2\pi)
  \end{array}\right)\xrightarrow{\mathbb{P}} 0
\end{align*}
as $n\to\infty$, where $Z_{n,U}(x):=\cos(Ux)W_{n,U}+\sin(Ux)V_{n,U}.$
\end{theorem}

\begin{proof}
The main technical step is provided by Lemma~\ref{linearized_no_convergence}, which treats the convergence of the linearized stochastic errors. The remainder terms are asymptotically negligible by Lemma~\ref{linearization_>0}.
To see the first line we set $x_1=x_2=0$, $w_1 \equiv 1$ and $w_2=w_\sigma^1$ in Lemma~\ref{linearized_no_convergence} and $w_U=w_\sigma^1$ in Lemma~\ref{linearization_>0}. The second and third line follow analogously. In order to derive the last line we observe
\begin{align*}
  \mathcal{F}^{-1}\left[\Delta\psi_{n}(Uu)
  w_\mu^{1}(u)\right](Ux)
  =2\int_0^1\mathrm{Re}(\Delta\psi_{n}(Uu)e^{-iUux})w_\mu^1(u)\mathrm{d}u/(2\pi)
\end{align*}
and apply Lemma~\ref{linearized_no_convergence} with $x_1=0$, $x_2=x$, $w_1 \equiv 1$ and $w_2=w_\mu^1$, . The remainder term vanishes by setting $w_U(u)=w_\mu^1(u)e^{-iUux}$ in Lemma~\ref{linearization_>0}.\qed
\end{proof}

The assumption $\mathcal{T}\in \mathcal{G}_s(R,\sigma_{\max})$ restricts $\sigma$ to the interval $[0,\sigma_{\max}]$.
The condition $\epsilon_n U_n^2\exp(T\sigma^2U_n^2/2) \to 0$ is especially fulfilled if $U_n\le\bar\sigma^{-1}(2\log(\epsilon_n^{-1})/T)^{1/2}$ for any $\bar\sigma>\sigma_{\max}$. For the estimation it suffices to know some upper bound $\sigma_{\max}$ of $\sigma$. The theorem shows that regardless whether one undersmoothes or not the stochastic errors converge with appropriate scaling to normal random variables.
For the statement on asymptotic normality we have to undersmooth and further knowledge on the volatility is necessary.

In many situations the volatility $\sigma$ is known or can be estimated easily.
The volatility is preserved under a change to an equivalent measure so that it
is the same for the risk--neutral and the real--world measure even if the price
process is only assumed to be a semimartingale. Then one of the methods for
volatility estimation from high frequency data in the presence of jumps can be
used to estimate the volatility.
Cont and Tankov~\cite{ContTankov} also need to fix the volatility for their calibration method of exponential L\'evy models in advance since their method chooses only among measures of L\'evy processes equivalent to a prior measure. They suggest using historical data or an earlier calibrated model for the choice of the prior and thus also of the volatility.
In the following we will assume either that the volatility $\sigma$ is known or that we have a sufficiently good estimator of the volatility.
To control the remainder term we choose $U_n$ such that $\epsilon_n U_n^2\exp(T\sigma^2U_n^2/2) \to 0$ as $n\to\infty$.
We also assume the undersmoothing condition
$\epsilon_n U_n^{s+1}\exp(T\sigma^2U_n^2/2) \to \infty$ as $n\to\infty$. A
smoothness parameter $s\ge2$ is implicitly assumed so that both conditions can
be satisfied simultaneously.
A possible choice of $U_n$ is
\begin{equation}\label{cut-off}
U_n:=\left(\frac{2}{T\sigma^2} \log\left(\frac{\epsilon_n^{-1}}{{\log(\epsilon_n^{-1})}^\alpha}\right)\right)^{1/2},
\end{equation}
where $\alpha\in(1,(s+1)/2)$. Then it holds
\begin{align*}
\epsilon_n U_n^{\beta}\exp(T\sigma^2 U_n^2/2)
\to \left\{\begin{array}{cc}\infty & \beta>2\alpha\\ \left(\frac{2}{T\sigma^2}\right)^\alpha & \beta=2\alpha  \\0 & \beta<2\alpha\end{array}\right.
\end{align*}
as $n\to\infty$.
Especially the term diverges for $\beta=s+1$ and converges to zero for $\beta=2$ so that both conditions on $U_n$ are fulfilled.
Next we state the joint asymptotic normality result for the severely ill--posed
case of positive volatility.

\begin{theorem}\label{asymptotic distribution >0}
Let $\sigma>0$ and $\rho\in L^\infty(\R)$. Let the cut--off value $U_n$ be
chosen such that $\epsilon_n U_n^2\exp(T\sigma^2U_n^2/2) \to 0$ and $\epsilon_n
U_n^{s+1}\exp(T\sigma^2U_n^2/2)\to\infty$ as $n\to \infty$.
Then
\begin{align*}
  &\frac1{\epsilon_n \exp(T\sigma^2U_n^2/2)}
  \left(\begin{array}{rcl}
  U_n^2\Delta\hat\sigma^2& - & d \,w_\sigma^1(1)W_{n,U_n}\\
  U_n\Delta\hat\gamma^{\phantom{2}} & - & d \,w_\gamma^1(1)V_{n,U_n}\\
  \Delta\hat\lambda^{\phantom{2}}  & - & d \,w_\lambda^1(1)W_{n,U_n}\\
  U_n^{-1}\Delta\hat\mu(0) & - & d \,w_0(1)W_{n,U_n}/(2\pi)\\
  U_n^{-1}\Delta\hat\mu(x) & - & d \,w_\mu^1(1)Z_{n,U_n}(x)/(2\pi)
  \end{array}\right)\xrightarrow{\mathbb{P}} 0
\end{align*}
as $n\to\infty$, where
$x\in\R\backslash\{0\}$, $Z_{n,U}(x):=\cos(Ux)W_{n,U}+\sin(Ux)V_{n,U}$ and $d$ is given by \eqref{d}.
\end{theorem}

\begin{proof}
  The undersmoothing condition $\epsilon_n
U_n^{s+1}\exp(T\sigma^2U_n^2/2)\to\infty$ yields for the cut--off value
  $U_n^{-(s+3)}=o(\epsilon_n U_n^{-2}\exp(T\sigma^2U_n^2/2))$ and thus the
approximation error of~$\hat\sigma^2$ vanishes. A similar reasoning applies to
the approximation errors of the other estimators. Since $\hat\sigma^2$ converges
with a faster rate than $\hat\gamma$ and $\hat \gamma$ with a faster rate than
$\hat\lambda$ the leading stochastic error terms are given in
Theorem~\ref{stochastic errors} and the convergence of the first three lines
follows by this theorem.
  For $x\neq0$ all stochastic errors in \eqref{error_mu} are negligible except the first one. We obtain the convergence in the last line by Theorem~\ref{stochastic errors}. We observe that $\mathcal{F}^{-1}[u w_\mu^1(u)](0)=0$, since $w_\mu^1$ is symmetric. For $x=0$ the relevant stochastic error terms are
  \begin{align*}
  &\phantom{\;=\;}\mathcal{F}^{-1}\left[\Delta\psi_{n}(Uu)
  w_\mu^{1}(u)\right](0)
  + \int_{0}^{1} \mathrm{Re}(\Delta\psi_{n}(Uu) )
  w_{\sigma}^{1}(u)\mathrm{d}u\, \mathcal{F}^{-1}\left[u^2w_\mu^{1}(u)\right](0)\\
  &\phantom{\;=\;}- 2\int_{0}^{1} \mathrm{Re}(\Delta\psi_{n}(Uu) )
  w_{\lambda}^{1}(u)\mathrm{d}u\, \mathcal{F}^{-1}\left[w_\mu^{1}(u)\right](0)\\
  &=\frac{1}{2\pi} 2\int_{0}^{1} \mathrm{Re}(\Delta\psi_{n}(Uu) )
  w_0(u)\mathrm{d}u.
  \end{align*}
  We apply Lemma~\ref{linearized_no_convergence} with $x_1=x_2=0$, $w_1\equiv1$ and $w_2=w_0$ to this term.
  The remainder term is asymptotically negligible by Lemma~\ref{linearization_>0}.
  This shows the convergence in the next to last line.\qed
\end{proof}

\subsection{Discussion of the results}

Theorems~\ref{asymptotic distribution =0} and~\ref{asymptotic distribution >0}
include the asymptotic distribution of $\hat\sigma^2$, which may be used for
testing the hypotheses $H_0:\sigma=\sigma_0$, see Section~6.2 in~\cite{soehl}.
If $\sigma$ is known, we can set $\hat\sigma^2=\sigma^2$. Then the statements of the theorems hold with $w_\sigma^1$ constant to zero. The estimation method can give negative values for $\hat\sigma^2$, $\hat\lambda$ and $\hat\nu(x)$. By a postprocessing step the estimated values can be corrected to be non--negative.

In Theorem~\ref{asymptotic distribution =0} the noise level $\rho$ enters only
locally into the asymptotic variance, whereas in Theorems~\ref{stochastic
errors}
and~\ref{asymptotic distribution >0} the asymptotic variance depends on the
$L^2$--norm of $\rho$ through the factor $d$. In fact for $\sigma=0$ it is
possible to estimate $\gamma$ and $\lambda$ directly from local properties of
the option function $\mathcal{O}$ at $\gamma T$ as remarked in
\cite{calibration}.
This local dependence on the noise level resembles some similarity to
deconvolution, for instance, to the case of ordinary smooth error densities
\cite{fan} or to the case of symmetric stable error densities whose
characteristic function decreases slower than the characteristic function of the
Cauchy distribution \cite{EsUh1}.
In both cases the density of the observations enters locally into the asymptotic
variance.
For the weight functions the local and global dependence is vice versa. In
Theorem~\ref{asymptotic distribution =0} the weight functions $w_\sigma^1$,
$w_\gamma^1$, $w_\lambda^1$, $w_0^1$ and $w_\mu^1$ enter globally into the
asymptotic variance while in Theorems~\ref{stochastic errors}
and~\ref{asymptotic distribution >0} only the values of the weight functions at
their endpoints appear in the asymptotic variance.

The asymptotic variance depends on the maturity. For positive volatility this dependence is through $d$ in \eqref{d}.
The martingale condition is equivalent to
the
equation $\sigma^2/2+\gamma-\lambda+\int_{-\infty}^{\infty}e^x\nu(x)\mathrm{d}
x=0$ , especially it holds that $\lambda-\gamma-\sigma^2/2\ge0$ with equality if
and only if $\lambda=0$ that is in the Black--Scholes case. In the case of
positive volatility $\sigma$ the asymptotic variance grows exponentially as
$T\to\infty$ if the jump intensity $\lambda$ is positive and it grows quadratic
as $T\to0$.

For $w_\sigma^1(1)$, $w_\gamma^1(1)$, $w_\lambda^1(1)$, $w_\mu^1(1)\in\R\backslash\{0\}$ Theorem~\ref{stochastic errors} describes the asymptotic distribution of the leading stochastic error term of $\hat\sigma^2$, $\hat\gamma$, $\hat\lambda$ and $\hat\mu(x)$, $x\neq0$, i.e., the other stochastic error terms are of smaller order. Theorem~\ref{asymptotic distribution >0} describes the asymptotic distribution of $\hat\mu$. Both theorems are for the case of positive volatility, where the noise in the frequency domain is exponentially heteroscedastic, so that the highest frequency, that is the cut--off frequency~$U$, dominates the stochastic error. Then this cut--off frequency $U$ can be seen in the asymptotic distribution of $\hat\mu$ through the oscillating process $Z_{n,U}$.
The variances in Theorems~\ref{stochastic errors} and~\ref{asymptotic distribution >0} converge by \eqref{W and V} and by the definition of $Z_{n,U}(x)$. If one only considers the stochastic errors of $\hat\sigma^2$, $\hat\gamma$, $\hat\lambda$ and $\hat\mu(0)$, then the covariances converge, too.
But for $x\neq0$ the covariance of the stochastic errors of $\hat\mu(x)$ and of
$\hat\sigma^2$ does not converge. The same holds for the covariance of the
stochastic errors of $\hat\mu(x)$ and $\hat\gamma$ as well as $\hat\mu(x)$ and
$\hat\lambda$. The phenomenon that the covariances do not convergence comes from
the fact that the stochastic error centers more and more at the cut--off
frequency. The sequence of cut--off values has a crucial influence on the
covariance. For estimators of the generalized distribution function of the
L\'evy density this is likely to lead to a similar dependence on the sequence of
cut--off values as observed in \cite{EsUh2005} for deconvolution with
supersmooth errors.

\section{Applications}\label{secConfidence}

\subsection{Construction of confidence intervals and confidence sets}

For $\sigma=0$ we define confidence intervals
\begin{align}\label{intervals sigma=0}
  \begin{array}{rll}
  I_{\gamma,n}:=&[\hat \gamma-\hat s_\gamma\epsilon_n U^{1/2} q_{\alpha/2},&\hat \gamma+\hat s_\gamma\epsilon_n U^{1/2} q_{\alpha/2}],\\
I_{\lambda,n}:=&[\hat \lambda-\hat s_\lambda\epsilon_n U^{3/2} q_{\alpha/2},&\hat \lambda+\hat s_\lambda\epsilon_n U^{3/2} q_{\alpha/2}],\\
I_{\mu(0),n}:=&[\hat \mu(0)-\hat s_{\mu(0)}\epsilon_n U^{5/2} q_{\alpha/2},&\hat \mu(0)+\hat s_{\mu(0)}\epsilon_n U^{5/2} q_{\alpha/2}],\\
I_{\mu(x),n}:=&[\hat \mu(x)-\hat s_{\mu(x)}\epsilon_n U^{5/2} q_{\alpha/2},&\hat \mu(x)+\hat s_{\mu(x)}\epsilon_n U^{5/2} q_{\alpha/2}],
  \end{array}
\end{align}
where $x\in\R\backslash\{0\}$, $q_\alpha$ denotes the $(1-\alpha)$--quantile of
the standard normal distribution and
\begin{align*}
\begin{array}{rl}
  \hat s_\gamma&:=\hat
s(0)\left(\int_0^1u^4w_\gamma^1(u)^2\mathrm{d}u\right)^{1/2},\\
  \hat s_\lambda&:=\hat
s(0)\left(\int_0^1u^4w_\lambda^1(u)^2\mathrm{d}u\right)^{1/2},\\
  \hat s_{\mu(0)}&:=\hat
s(0)\left(\int_0^1u^4w_0(u)^2\mathrm{d}u\right)^{1/2}/(2\pi),\\
  \hat
s_{\mu(x)}&:=\hat s(x)
\left(\int_0^1u^4w_\mu^1(u)^2\mathrm{d}u\right)^{1/2}/(2\pi),
\end{array}
\end{align*}
with $\hat s
(x):=2\sqrt{\pi}\rho(x+T\hat\gamma)\exp(T(\hat\lambda-\hat\gamma))/T$.
We fix some arbitrarily slowly decreasing function $h$ with $h(u)\to0$ as $|u|\to\infty$.
We denote by $\mathcal{H}_s(R,\sigma_{\max})$ the subset of L\'evy triplets in $\mathcal{G}_s(R,\sigma_{\max})$ that satisfy in addition
\begin{align}\label{H}
  \|\mathcal{F}\mu\|_\infty\le R,\quad |\mathcal{F}\mu(u)|\le R\, h(u), \quad \forall u\in\R.
\end{align}
The additional conditions are used to extend the convergence in the theorems to
be uniform over all L\'evy triplets in~$\mathcal{H}_s(R,\sigma_{\max})$, see
Theorem~5.1 in~\cite{soehl}, and to obtain \emph{honest} confidence sets meaning
that the level is achieved uniformly over a class of L\'evy triplets.
\begin{corollary}
Let $\sigma=0$. On the assumptions of Theorem~\ref{asymptotic distribution =0}
and on the assumption that $\rho$ is positive and continuous
\begin{align*}
\lim_{n\to\infty}\,\inf_{\mathcal{T}\in\mathcal{H}_s(R,0)}\mathbb{P}_\mathcal{T}
(\vartheta\in I_{\vartheta,n})=
1-\alpha
\end{align*}
holds for the intervals \eqref{intervals sigma=0} and for all
$\vartheta\in\{\gamma,\lambda,\mu(x)|x\in\R\}$.
\end{corollary}
If the infimum in the corollary is omitted, then the statement holds for all
L\'evy triplets $\mathcal{T}$ in $\mathcal{G}_s(R,0)$ and is a direct
consequence of Theorem~\ref{asymptotic distribution =0}. The same holds for the
other confidence intervals and sets, where in the case of positive volatility
the statements hold for the corresponding L\'evy triplets $\mathcal{T}$ in
$\mathcal{G}_s(R,\sigma_{\max})$ and follow from Theorem~\ref{asymptotic
distribution >0}.
\begin{remark}
To consider the 
two parameters $\gamma$ and $\lambda$ jointly, we define the confidence set
$A_n:=\{(\hat\gamma+\epsilon_n U^{1/2}\hat{s}_\gamma x,\hat\lambda+\epsilon_n
U^{3/2}\hat{s}_\lambda y)^\top|\:x^2+y^2\le k_\alpha\}$, where $k_\alpha$
denotes the $(1-\alpha)$--quantile of the chi--squared distribution $\chi_2^2$
with two degrees of freedom.
  Then it holds
  \begin{align*}
  \lim_{n\to\infty}\,\inf_{\mathcal{T}\in\mathcal{H}_s(R,0)}\mathbb{P}_\mathcal{T}((\gamma,\lambda)^\top \in A_{n})=1-\alpha.
  \end{align*}
\end{remark}

For $\sigma>0$ we define confidence intervals
\begin{align}\label{intervals sigma>0}
  \begin{array}{rll}
  I_{\gamma,n}:=&[\hat\gamma-\hat s_\gamma\epsilon_n U^{-1}e^{T\sigma^2U^2/2}q_{\alpha/2},&\hat\gamma+\hat s_\gamma\epsilon_n U^{-1}e^{T\sigma^2U^2/2}q_{\alpha/2}],\\
  I_{\lambda,n}:=&[\hat\lambda-\hat s_\lambda\epsilon_n e^{T\sigma^2U^2/2}q_{\alpha/2},&\hat\lambda+\hat s_\lambda\epsilon_n e^{T\sigma^2U^2/2}q_{\alpha/2}],\\
  I_{\mu(0),n}:=&[\hat\mu(0)-\hat s_{\mu(0)}\epsilon_n Ue^{T\sigma^2U^2/2}q_{\alpha/2},&\hat\mu(0)+\hat s_{\mu(0)}\epsilon_n Ue^{T\sigma^2U^2/2}q_{\alpha/2}],\\
  I_{\mu(x),n}:=&[\hat\mu(x)-\hat s_{\mu}\epsilon_n Ue^{T\sigma^2U^2/2}q_{\alpha/2},&\hat\mu(x)+\hat s_{\mu}\epsilon_n Ue^{T\sigma^2U^2/2}q_{\alpha/2}],
  \end{array}
\end{align}
where $x\in\R\backslash\{0\}$,
\begin{align*}
  \left(\begin{array}{l}\hat s_\gamma\\\hat s_\lambda\\\hat s_{\mu(0)}\\\hat s_{\mu}\end{array}\right)
  :=\frac{\sqrt{2}\|\rho\|_{L^2(\R)}}
{\exp(T(\sigma^2/2+\hat\gamma-\hat\lambda))T^2\sigma^2}
  \left(\begin{array}{l}
  |w_\gamma^1(1)|\\
  |w_\lambda^1(1)|\\
  |w_0(1)|/(2\pi)\\
  |w_\mu^1(1)|/(2\pi)
  \end{array}\right)
\end{align*}
and $q_\alpha$ denotes the $(1-\alpha)$--quantile of the standard normal distribution. We assume that the weight functions are chosen such that $w_\gamma^1(1)$, $w_\lambda^1(1)$, $w_0(1)$, $w_\mu^1(1)\in\R\backslash\{0\}$.
We note that instead of estimating $\rho$ nonparametrically it suffices for
positive volatility to estimate the $L^2$--norm of $\rho$.
For example, in the case of equidistant design we can first estimate
$\mathcal{O}$ with the standard Nadaraya--Watson estimator for regression and
then estimate the $L^2$--norm of $\rho$ from the sum of the
squared residuals, which
leads to a consistent estimator as shown by Dette and Neumeyer
\cite{DetteNeumeyer2001}.

\begin{corollary}
Let $\sigma>0$. On the assumptions of Theorem~\ref{asymptotic distribution >0}
\begin{align*}
\lim_{n\to\infty}\inf_{\mathcal{T}}\mathbb{P}_\mathcal{T}(\vartheta\in
I_{\vartheta,n})=
1-\alpha
\end{align*}
holds for the intervals \eqref{intervals sigma>0} and for all
$\vartheta\in\{\gamma,\lambda,\mu(x)|x\in\R\}$, where the infimum is over all
$\mathcal{T}\in\mathcal{H}_s(R,\sigma_{\max})$ with volatility $\sigma$.
\end{corollary}
For the pair $(\gamma,\lambda)^\top$ a uniform confidence set may be obtained
similarly as in the case $\sigma=0$. Since for $x\in\R\backslash\{0\}$ the
covariance of $Z_{n,U_n}(x)$ and $V_{n,U_n}$ and the covariance of
$Z_{n,U_n}(x)$ and $W_{n,U_n}$ do not converge, confidence sets for
$(\gamma,\mu(x))^\top$ and $(\lambda,\mu(x))^\top$ have to be constructed
differently.
Let us illustrate how to proceed in this case by constructing a confidence set for $(\mu(x_1),\mu(x_2))^\top$, $x_1,x_2\in\R\backslash\{0\}$. By Theorem~\ref{asymptotic distribution >0} the convergence
\begin{align*}
  &\frac1{\epsilon_n \exp(T\sigma^2U_n^2/2)}
  \left(\frac{1}{U_n} \left(\begin{array}{c}
  \Delta\hat\mu(x_1) \\
  \Delta\hat\mu(x_2)
  \end{array}\right)
  -\frac{d \, w_\mu^1(1)}{2\pi}\left(\begin{array}{rcl}
  Z_{n,U_n}(x_1)\\
  Z_{n,U_n}(x_2)
  \end{array}\right)\right)\xrightarrow{\mathbb{P}} 0
\end{align*}
holds for $n\to \infty$. We
define
\begin{align*}
  M_n:=\left(\begin{array}{cc}\cos(U_nx_1)&\sin(U_nx_1)\\ \cos(U_nx_2)&\sin(U_nx_2)\end{array}\right),
\end{align*}
and observe that the components of
$M_n^{-1}$ are bounded for all $n$ for which the absolute value of the
determinant
is bounded from below by some $c>0$, i.e., we have $|\sin(U_n(x_2-x_1))|\ge c$.
For such $n$
\begin{align*}
  &\frac1{\epsilon_n \exp(T\sigma^2U_n^2/2)}
  \left(\frac{M_{n}^{-1}}{U_n} \left(\begin{array}{c}
  \Delta\hat\mu(x_1) \\
  \Delta\hat\mu(x_2)
  \end{array}\right)
  -\frac{d \, w_\mu^1(1)}{2\pi}\left(\begin{array}{rcl}
  W_{n,U_n}\\
  V_{n,U_n}
  \end{array}\right)\right)\xrightarrow{\mathbb{P}} 0
\end{align*}
holds for $n\to\infty$.
We apply the additive version of Slutsky's lemma together with the convergence~\eqref{W and V} of the appropriately scaled random variables $W_{n,U_n}$ and $V_{n,U_n}$. In view of the definition of $d$ in~\eqref{d} we observe that $\hat s_\mu$ is a consistent estimator of $d |w_\mu^1(1)|/(2\pi)$ and we apply the multiplicative version of Slutsky's lemma, which then leads to
\begin{align*}
  \frac{1}{\hat s_\mu \epsilon_n U_n \exp(T\sigma^2U_n^2/2)}M_{n}^{-1}
  \left(\begin{array}{c}
  \Delta\hat\mu(x_1) \\
  \Delta\hat\mu(x_2)
  \end{array}\right)
  \xrightarrow{d} \left(\begin{array}{c}W\\V\end{array}\right)
\end{align*}
for $n\to\infty$ such that $|\sin(U_n(x_2-x_1))|\ge c$. We define
\[
  B_n:=\left(\hat\mu(x_1), \hat\mu(x_2)\right)^\top +M_n \{\hat s_\mu \epsilon_n U_n \exp(T\sigma^2U_n^2/2) (x,y)^\top|\,x^2+y^2\le k_\alpha\},
\]
where $k_\alpha$ denotes the $(1-\alpha)$--quantile of the chi--squared distribution $\chi^2_2$ with two degrees of freedom.
Then
\begin{align*}
\lim_{\substack{|\sin(U_n(x_2-x_1))|\ge c\\ n\to\infty}}
\mathbb{P}_\mathcal{T}((\mu(x_1),\mu(x_2))^\top \in B_{n})=1-\alpha
\end{align*}
holds for all $\mathcal{T}\in\mathcal{G}_s(R,\sigma_{\max})\cap\{\sigma>0\}$.

\subsection{A numerical example}

We consider the \emph{Merton} jump diffusion model \cite{merton:1976}, where the jump density is specified by
\[\nu(x)=\frac{\lambda}{\sqrt{2\pi}\,\zeta}\exp\left(-\frac{(x-\eta)^2}{2\zeta^2
} \right) , \quad x\in\R,\]
with parameters $\sigma,\lambda\ge0$, $\zeta>0$, $\eta\in\R$ and where
$\gamma\in \R$ is determined by the martingale condition. We simulate data with
the parameters $\sigma=0.1$, $\lambda=5$, $\eta=-0.1$, $\zeta=0.2$, which
imply $\gamma=0.379$. The interest rate is taken to be $r=0.06$. We observe
prices of $n=100$ European options with maturity $T=0.25$. The strike prices are
obtained from sampling the data points $(x_j)$ from a centered normal
distribution with variance $1/2$, so that more strike prices are sampled at the
money than in or out of the money. The observation error is chosen to be a
centered normal distribution with variance $\delta^2\mathcal{O}(x_j)^2$,
$\delta=0.01$. Belomestny and Rei\ss~\cite{BelomestnyReiss2006b} describe the
implementation of the estimation method in detail.

\ifodd\switch
\begin{figure}
\centering
\includegraphics[width=0.6\textwidth]{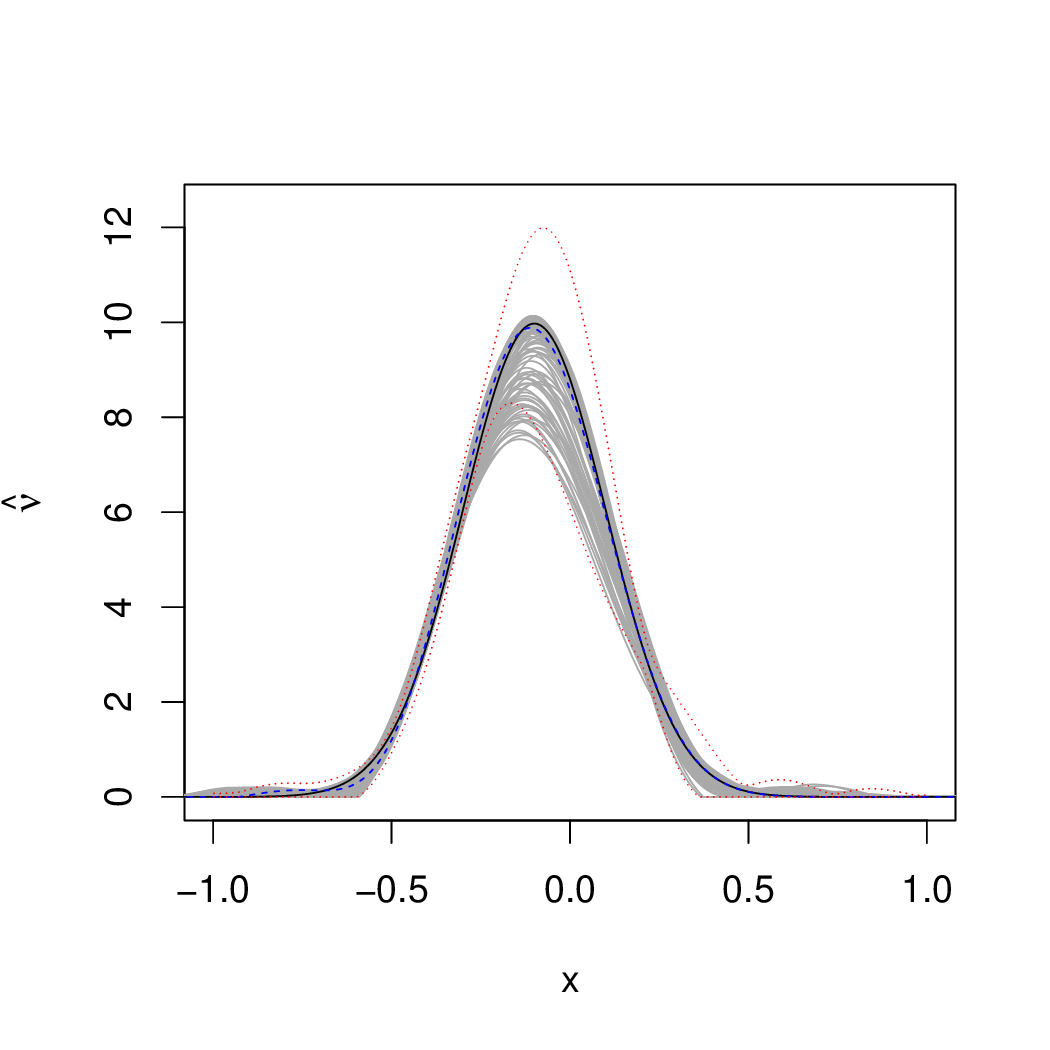}
\caption{True (black, solid) and estimated (blue, dashed) L\'evy density with
pointwise 95\% confidence intervals (red, dotted), using the oracle cut--off
value $U=18.5$. Additional 100 estimated L\'evy densities (gray) from a Monte
Carlo simulation of the Merton model.}
\label{fig:1}
\end{figure}
\else
\begin{figure}[ht]
\centering
\includegraphics[width=0.6\textwidth]{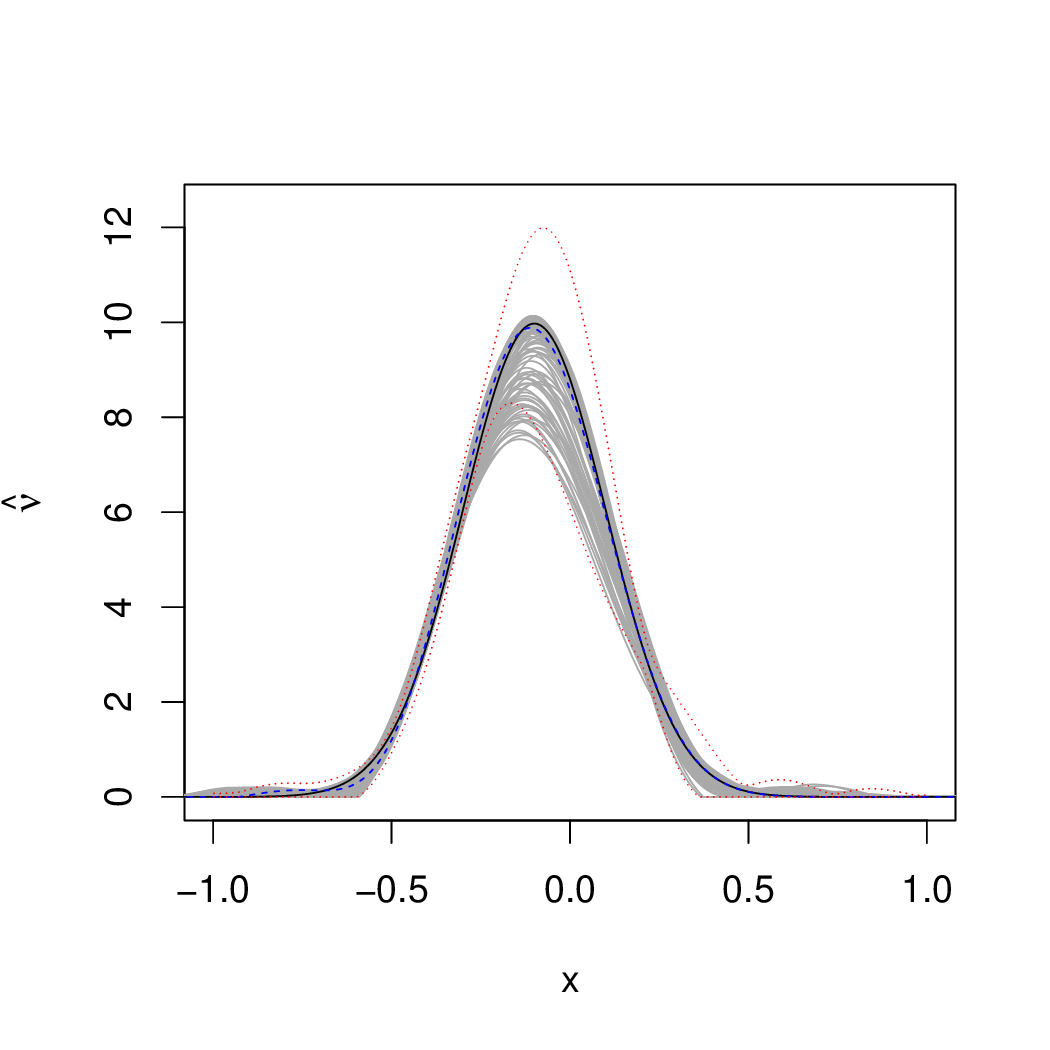}
\caption{True (black, solid) and estimated (blue, dashed) L\'evy density with
pointwise 95\% confidence intervals (red, dotted), using the oracle cut--off
value $U=18.5$. Additional 100 estimated L\'evy densities (gray) from a Monte
Carlo simulation of the Merton model.}
\label{fig:1}
\end{figure}
\fi

We interpolate the corresponding European call prices linearly. The weight functions are chosen as in \cite{soehltrabs} with smoothness parameter $s=2$. In the simulations the confidence intervals based on the asymptotic distribution turn out to be to conservative.
Such confidence intervals would be based on the asymptotic variance of the linearized stochastic errors that is the stochastic errors, where the linearization \eqref{eqLinearization} is used. Instead of taking the asymptotic variance of the linearized stochastic errors, we derive confidence intervals from the finite sample variance \eqref{covariance} of the linearized stochastic errors.
In the finite sample variance we substitute $\sigma$, $\gamma$, $\lambda$, $\mu$ and $\nu$ by their respective estimators. This yields feasible confidence intervals. We estimate with the oracle choice of the cut--off values and perform 1000 Monte Carlo iterations. The coverage probabilities of $95\%$ confidence intervals for $\sigma^2$, $\gamma$ and $\lambda$ are approximately 0.98, 0.95 and 0.91, respectively.
We see that the that the coverage probabilities are close to the prescribed confidence level.

Figure~\ref{fig:1} illustrates the true and an estimated L\'evy density, and the pointwise $95\%$ confidence intervals based on the finite sample variance. Additional 100 estimated L\'evy densities are plotted. At zero there is a negative bias since the peak is smoothed out. The precise construction of the confidence intervals and more calibration results are contained in~\cite{soehltrabs}.

\section{Conclusion}\label{secConclusion}

We have shown asymptotic normality in a nonparametric calibration method for exponential L\'evy models. These results were used to derive confidence intervals and confidence sets. We have seen in a numerical example that confidence intervals based on finite sample variance perform well in terms of coverage probabilities. The confidence intervals extend the calibration method beyond a pure point estimate and enable an assessment of the calibration error.
Although parametric models might be fitted better, the parametric approach is always exposed to the risk of model misspecification and the obtained confidence results should be used for a goodness--of--fit test.

The estimation method and the asymptotic normality results may be adapted to other models as long as there is an equation relating the option function to the characteristic function and the parameters of interest appear in the characteristic function. The constructed confidence intervals and sets may be used to quantify the errors in pricing, hedging and risk management.

\section{Auxiliary results and remaining proofs}\label{secProofs}

We will now state the conditions more precisely on which the regression
model~\eqref{regression} and the Gaussian white noise model~\eqref{whitenoise}
are equivalent. As a simplification we assume that the observations in the
regression model are equidistant with mesh size $\Delta_n$.
We restrict the Gaussian white noise model to a sequence of growing intervals $[x_1-\Delta_n,x_n+\Delta_n]$.
We suppose $\rho^2>0$ to be an absolutely continuous function and
$|\frac{\partial}{\partial x}\log\rho(x)|\le C$ to hold for some $C<\infty$.
The functions $\mathcal{O}$ are both uniformly bounded
$\mathcal{O}(x)=S^{-1}\mathcal{C}(x,T)-(1-e^x)^+\le 1$ as well as uniformly
Lipschitz
$|\mathcal{O}'(x)|=|\int_{-\infty}^{x}\mathcal{O}''(x)\mathrm{d}x-\mathbf{1}_{\{
x>0\}} +e^{(\gamma-\lambda)T}\mathbf{1}_{\{x>\gamma T, \sigma=0\}}|\le 4+
e^{RT}$, where we used Proposition~2.1 in~\cite{calibration} and $|\gamma|\le
R$. These properties of $\mathcal{O}$ are used to apply Corollary~4.2
in~\cite{BrownLow}.

\subsection{Proof of Proposition~\ref{sup L}}
First we define $X(u):=\int_{-\infty}^{\infty}e^{iux}\rho (x) \mathrm{d}W(x)$.
Since $X(-u)=\overline{X(u)}$ it suffices to consider suprema of the absolute value $|X(u)|$ over positive index sets.
We assumed that there is an $p>1$ such that
$\int_{-\infty}^{\infty}(1+|x|)^{p}\rho(x)^2\mathrm{d}x<\infty$. It is shown
in \cite{polar} that there exists a number $c>0$ such that
$\sqrt{\mathbb{E}[|X(u)-X(v)|^2]}\le c|u-v|^\alpha$  for all $u,v\in\R$ with
$\alpha:=\min(p/2,1)\in(1/2,1]$. Denote by $N_{\delta}(I,r)$ the covering
number, that is the minimum number of closed balls of radius $r$ in the
metric~$\delta$ with centers in $I$ that cover $I$. We define
$\delta(u,v):=c|u-v|^\alpha$ and $d(u,v):=\sqrt{\mathbb{E}[|X(u)-X(v)|^2]}$. A
ball of radius~$r$ in the metric $\delta$ covers an interval of length
$2(r/c)^{1/\alpha}$.
Thus, it holds
\[N_{\delta}([0,U],r) =\left\lceil U
\left(c/r\right)^{1/\alpha}/2\right\rceil.\]
The radius of the smallest ball with center in $[0,U]$ that contains $[0,U]$ is
$c (U/2)^\alpha$ with respect to the metric~$\delta$. There exists $D<\infty$
such that $d(u,v)\le D$ for all $u,v\in\R$. For $U$ large enough such that
$cU^\alpha\ge D$ we have the entropy bound
\begin{align}
J([0,U],d)
&:= \int_0^{\infty}\left(\log(N_d([0,U],r))\right)^{1/2}\mathrm{d}r
\le
\int_0^{D}\left(\log(N_\delta([0,U],r))\right)^{1/2}\mathrm{d}r\label{entropy}\\
&
\le \int_0^{D}\left(\log\left(U \left(c/r\right)^{1/\alpha}\right)\right)^{1/2}\mathrm{d}r
\le  \alpha^{-1/2} \int_0^{D}\left(\log\left(  U^\alpha c/r\right)\right)^{1/2}\mathrm{d}r,\nonumber\\
\intertext{here we substitute $r= U^\alpha c s$,}
& \le  \alpha^{-1/2}  U^\alpha c \int_0^{D/(U^\alpha c)}\left(\log\left( 1/s\right)\right)^{1/2}\mathrm{d}s.\label{entropy_bound}
\end{align}
This integral is solved by
\[\int_0^x\sqrt{\log y^{-1}}\mathrm{d}y=\frac{\sqrt{\pi}}{2} -\frac{\sqrt{\pi}}{2}\text{Erf}(\sqrt{\log x^{-1}}) +x\sqrt{\log x^{-1}},\]
where $\text{Erf}(y)=\frac{2}{\sqrt{\pi}}\int_0^ye^{-t^2}\mathrm{d}t$. For all $y>0$ the estimate \(1-\text{Erf}(y)<\exp(-y^2)/(\sqrt{\pi}y)\) holds. For each $\alpha>0$ this yields $\tilde c>0$ such that for all $x\in(0,1/2^\alpha)$
\[\int_0^x\sqrt{\log y^{-1}}\mathrm{d}y\le\tilde c x\sqrt{\log x^{-1}}.\]
Thus, \eqref{entropy_bound} can be bounded by
\begin{align*}
(\alpha^{-1/2} \tilde c D) \sqrt{\log({U}^\alpha c/D)}\lesssim \sqrt{\log(U)}
\end{align*}
as $U\to\infty$.
Consequently $\sqrt{\log(U)}$ is an asymptotic upper bound of the entropy integral~\eqref{entropy}. We apply Dudley's theorem \cite[p. 219]{kahane1985} to the real part of $X$.
For all $q\ge1$ this yields a continuous version $X'$ of $\mathrm{Re}(X)$ with
\begin{equation}\label{sup X}
\mathbb{E}\left[\sup_{u\in[0,U]}|X'(u)|^q\right]\lesssim(\log(U))^{q/2}
\end{equation}
as $U\to\infty$.
Since $X'$ and $\mathrm{Re}(X)$ are both continuous they are indistinguishable and \eqref{sup X} holds for $\mathrm{Re}(X)$ likewise. We obtain analogously for all $q\ge1$
\begin{align*}
\mathbb{E}\left[\sup_{u\in[0,U]}|\mathrm{Im}(X(u))|^q\right]
\lesssim(\log(U))^{q/2}
\end{align*}
as $U\to\infty$.
We estimate from above for all $q\ge1$
\begin{align}
&\phantom{\;\le\;}\mathbb{E}\left[\sup_{u\in[-1,1]}|\mathcal{L}_{n,U}(u)|^q\right] \nonumber\\
&\le \sup_{u\in[-U,U]}\left|\frac{\epsilon_n iu(1+iu)}{T(1+iu(1+iu)\mathcal{FO}(u))}\right|^q \mathbb{E}\left[\sup_{u\in[0,U]}|X(u)|^q\right]\nonumber\\
&\le \left(\frac{\epsilon_n U\sqrt{1+U^2} } {T\exp(T(-\sigma^2U^2/2+\sigma^2/2+\gamma-\lambda-\|\mathcal{F}\mu\|_\infty))}\right)^q
\mathbb{E}\left[\sup_{u\in[0,U]}|X(u)|^q\right]\nonumber\\
&\lesssim \left(\epsilon_n U^2 \sqrt{\log(U)} \exp(T\sigma^2U^2/2)\right)^q\label{bound sup L}
\end{align}
as $U\to\infty$.
This completes the proof for the case $\sigma=0$. For $\sigma>0$ we observe
\begin{align*}
\mathbb{E}\left[\sup_{u\in[-1,1]}|\mathcal{L}_{n,U}(u)|^q\right]
\le \mathbb{E}\left[\sup_{|u|\le U-1}|\mathcal{L}_{n,1}(u)|^q\right] +\mathbb{E}\left[\sup_{|u|\in[U-1,U]}|\mathcal{L}_{n,1}(u)|^q\right].
\end{align*}
By the previous considerations the growth of the first part can be bounded by
\begin{align*}
\left(\epsilon_n (U-1)^2 \sqrt{\log(U-1)} \exp(T\sigma^2(U-1)^2/2)\right)^q
\lesssim \left(\epsilon_n U^2  \exp(T\sigma^2U^2/2)\right)^q.
\end{align*}
For the second part we note that as in \eqref{entropy} we have
\begin{align*}
J([U-1,U],d)\le
\int_0^{D}\left(\log(N_\delta([U-1,U],r))\right)^{1/2}\mathrm{d}r
= \int_0^{D}\left(\log(N_\delta([0,1],r))\right)^{1/2}\mathrm{d}r
\end{align*}
and thus the entropy does not depend on $U$. For $u\in[U-1,U]$ the process $X$ does not contribute a logarithmic factor and it holds
\begin{align*}
\mathbb{E}\left[\sup_{u\in[-1,1]}|\mathcal{L}_{n,U}(u)|^q\right]
\lesssim \left(\epsilon_n U^2  \exp(T\sigma^2U^2/2)\right)^q
\end{align*}
as $U\to\infty$.

\subsection{The linearized stochastic errors}

The linearized stochastic errors are of the form $\int_{0}^{1}f_j(u)\mathcal{L}_{n,U}(u)\mathrm{d}u$, where $f_j$ with $j=1,\dots,n$ are Riemann--integrable function in $L^\infty([0,1])$. Next we will show that these are jointly normal distributed.
Almost surely $\mathcal{L}_{n,U}$ is continuous. Thus, almost surely the $f_j \mathcal{L}_{n,U}$ are Riemann--integrable and almost surely
\[\frac{1}{m}\sum_{k=1}^{m} f_j(k/m)\mathcal{L}_{n,U}(k/m) \rightarrow\int_0^1 f_j(u)\mathcal{L}_{n,U}(u)\mathrm{d}u\]
as $m\to\infty$. Let $C>0$ be such that $\|f_j\|_\infty\le C$ for all $j=1,\dots,n$.
For each $m$ the $n$ sums are joint, centered normal random variables. For $m\to \infty$ the covariance matrix converges by the dominated convergence theorem with the dominating function $C^2\sup_{u\in[0,1]}|\mathcal{L}_{n,U}(u)|^2$, where $\sup_{u\in[0,1]}|\mathcal{L}_{n,U}(u)|^2$ is an integrable random variable by Proposition~\ref{sup L}. Thus, the characteristic function converges pointwise. By L\'evy's continuity theorem this shows that the sums convergence jointly in distribution to normal random variables. So $\int_{0}^{1}f_j(u)\mathcal{L}_{n,U}(u)\mathrm{d}u$ are jointly normal distributed.

For a fixed cut--off value $U$ the linearized stochastic errors are jointly normal distributed. So the natural question is whether the appropriately scaled covariance matrix converges for $U\to\infty$.

Let $w_j,w_k:[0,1]\to\R$ be Riemann--integrable functions in $L^\infty([0,1])$. It holds
\begin{align*}
&T\int_0^1 w_j(u) \mathcal{L}_{n,U}(u)e^{-iUux_j}\mathrm{d}u\\
&={\epsilon_n U^2e^{-T(\sigma^2/2+\gamma-\lambda)}}\int_0^1
f_U(u)\int_{-\infty}^\infty
e^{iUu(x-\theta_j)+T\sigma^2U^2u^2/2}\rho(x)\mathrm{d}W(x)\mathrm{d}u,
\end{align*}
where
\begin{equation}\label{f}
f_U(u):=\frac{w_j(u)(-u^2+iu/U)}{\exp(T\mathcal{F}\mu(Uu))}
\end{equation}
and $\theta_j:=x_j+T\sigma^2+T\gamma$. We define analogously
\begin{equation}\label{g}
g_U(u):=\frac{w_k(u)(-u^2+iu/U)}{\exp(T\mathcal{F}\mu(Uu))}
\end{equation}
and $\theta_k:=x_k+T\sigma^2+T\gamma$. We extend $f_U$ and $g_U$ by zero outside the interval $[0,1]$.

Since $\mathbb{E}\left[\sup_{u\in[-1,1]}|\mathcal{L}_{n,U}(u)|^2\right]<\infty$ we may apply Fubini's theorem and then we apply the It\^{o} isometry to obtain
\begin{align}
&{T^2e^{2T(\sigma^2/2+\gamma-\lambda)}} \mathbb{E}\left[\int_0^1w_j(u)\mathcal{L}_{n,U}(u)e^{-iUux_j}\mathrm{d}u \overline{\int_0^1w_k(v)\mathcal{L}_{n,U}(v)e^{-iUvx_k}\mathrm{d}v}\right]\nonumber\\
&={\epsilon_n^2U^4}\int_0^1\int_0^1 \int_{-\infty}^\infty f_U(u)e^{iUu(x-\theta_j)+T\sigma^2U^2u^2/2}\nonumber\\
&\phantom{\;=\;}\overline{g_U(v)e^{iUv(x-\theta_k)+T\sigma^2U^2v^2/2}}\rho(x)^2
\mathrm{d}x \mathrm{d}u\mathrm{d}v.\label{ito}
\end{align}
To separate real and imaginary part we will also need
\begin{align}
&{T^2e^{2T(\sigma^2/2+\gamma-\lambda)}} \mathbb{E}\left[\int_0^1w_j(u)\mathcal{L}_{n,U}(u)e^{-iUux_j}\mathrm{d}u {\int_0^1w_k(v)\mathcal{L}_{n,U}(v)e^{-iUvx_k}\mathrm{d}v}\right]\nonumber\\
&={\epsilon_n^2U^4}\int_0^1\int_0^1 \int_{-\infty}^\infty f_U(u)e^{iUu(x-\theta_j)+T\sigma^2U^2u^2/2}\nonumber\\
&\phantom{\;=\;}{g_U(v)e^{iUv(x-\theta_k)+T\sigma^2U^2v^2/2}}\rho(x)^2\mathrm{d}
x \mathrm{d}u\mathrm{d}v.\label{ito_II}
\end{align}

\begin{lemma}\label{linearized_0}
Let $\sigma=0$.
For $j=1,\dots,m $ take $x_j\in\R$ and $w_j:[0,1]\to\R$ to be
Riemann--integrable functions in $L^\infty([0,1])$.
Let $\rho$ be continuous at the points $x_1+T\gamma$,
$x_2+T\gamma,\dots,x_m+T\gamma$ and
let $\mathcal{F}\rho^2\in L^1(\R)$.
Let $W_{x_1},\dots,W_{x_m},V_{x_1},\dots,V_{x_m}$ be Brownian motions.
If $x_j=x_k$ let $W_{x_j}=W_{x_k}$ and $V_{x_j}=V_{x_k}$ otherwise let the Brownian motions be distinct. Let the set $\left\{W_{x_1},\dots,W_{x_m},V_{x_1},\dots,V_{x_m}\right\}$ consist of independent Brownian motions.
Then
\[\frac1{\epsilon_n U^{3/2}}\int_{0}^{1}w_j(u)\mathcal{L}_{n,U}(u)e^{-iUux_j}\mathrm{d}u\] converge jointly in distribution to
\[\frac{\sqrt{\pi}\rho(x_j+T\gamma)}{T\exp(T(\gamma-\lambda))}
\left(\int_0^{1}u^2w_j(u)\mathrm{d}W_{x_j}(u)+i\int_0^{1}u^2w_j(u)\mathrm{d}V_{
x_j}(u)\right)\]
as $U\rightarrow \infty$.
\end{lemma}

\begin{proof}
We will first consider the case $x_j=x_k$.
We have seen that
\begin{align}
&{T^2e^{2T(\gamma-\lambda)}}\mathbb{E}\left[\int_0^1w_j(u)\mathcal{L}_{n,U}(u)e^{-iUux_j}\mathrm{d}u \overline{\int_0^1w_k(v)\mathcal{L}_{n,U}(v)e^{-iUvx_j}\mathrm{d}v}\right]\nonumber\\
&={\epsilon_n^2U^4} \int_{-\infty}^\infty \mathcal{F}f_U(U(x-\theta_j))
\overline{\mathcal{F}g_U(U(x-\theta_j))}\rho(x)^2\mathrm{d}x,\nonumber\\
\intertext{where $f_U$ and $g_U$ are defined as in \eqref{f} and \eqref{g}, respectively, and $\theta_j=x_j+T\gamma$,}
&={\epsilon_n^2U^3} \int_{-\infty}^\infty
\mathcal{F}f_U(y)\overline{\mathcal{F}g_U(y)\rho(y/U+\theta_j)^2}\mathrm{d}
y.\nonumber\\
\intertext{We notice that $\mathcal{F}\rho^2\in L^1(\R)$ implies $\rho^2\in
L^\infty(\R)$ and we obtain by the Plancherel identity}
&=2\pi\epsilon_n^2U^3\int_0^1 f_U(u)
\overline{(g_U(v) *
\mathcal{F}^{-1}(\rho(y/U+\theta_j)^2)(v))(u)}\mathrm{d}u,\label{convolution_I}
\end{align}
since the support of $f_U$ is $[0,1]$.
Because we are only interested in the limit $U\to\infty$ we may assume $U\ge1$.
By the Riemann--Lebesgue lemma $\mathcal{F}\mu(u)$ tends to zero as
$|u|\rightarrow \infty$. The factor $f_U(u)$ converges for each $u\in[0,1]$ to
$-u^2w_j(u)$ as $U\to\infty$ and the functions are dominated by a constant
independent of $U$. In order to apply dominated convergence it suffices that the
second factor is dominated by a constant independent of $U$ and converges
stochastically with respect to the Lebesgue measure on $\R$. We have
\begin{align*}
(g_U(v) * \mathcal{F}^{-1}(\rho(y/U+\theta_j)^2)(v))(u)
=\int_{-\infty}^\infty g_U(u-v)
\mathcal{F}^{-1}(\rho(y+\theta_j)^2)(Uv)U\mathrm{d}v
\end{align*}
By assumption $\mathcal{F}\rho^2$ lies in $L^1(\R)$ and so does
$\mathcal{F}^{-1}(\rho(y+\theta_j)^2)$. A dominating constant is
$\sqrt{2}\|w_k\|_\infty
\exp(T\|\mathcal{F}\mu\|_\infty)\,\|\mathcal{F}^{-1}(\rho(y+\theta_j)^2)\|_{
L^1(\R)}$.
It holds
\begin{align}
  \int_{-\infty}^\infty \mathcal{F}^{-1}(\rho(y+\theta_j)^2)(Uv)U\mathrm{d}v
  &=\int_{-\infty}^\infty \mathcal{F}^{-1}(\rho(y+\theta_j)^2)(v)\mathrm{d}v
\nonumber\\ &=\mathcal{F}\mathcal{F}^{-1}(\rho(y+\theta_j)^2)(0)
=\rho(\theta_j)^2.\label{multiple}
\end{align}
$\delta_U(v):=\mathcal{F}^{-1}(\rho(y+\theta_j)^2)(Uv)U$ is the multiple of what
is called approximate identity or nascent delta function. The basic theorem on
approximate identities states that $h*\delta_n$ converges to $h$ in $L^1(\R)$ as
$n\to\infty$ for $h\in L^1(\R)$.
Thus,
\[(-v^2w_k(v)\mathbf{1}_{[0,1]}(v))*\delta_U(v)(u)
\to -u^2w_k(u)\mathbf{1}_{[0,1]}(u)\rho(\theta_j)^2\quad\text{ in }L^1(\R)\]
as $U\to\infty$ \cite[p. 28]{grafakos2004} and in particular stochastically. If
$u\neq0$, then there is a neighborhood of $u$ where
$g_U(u)+u^2w_k(u)\mathbf{1}_{[0,1]}(u)$ converges uniformly to zero. The term
$(g_U(v)+v^2w_k(v)\mathbf{1}_{[0,1]}(v))*\delta_U(v)(u)$ converges to
zero almost surely
and in particular stochastically.
Therefore, $g_U(v)*\delta_U(v)(u)$ converges to
$-u^2w_k(u)\mathbf{1}_{[0,1]}(u)\rho(\theta_j)^2$ stochastically with respect to
the Lebesgue measure on $\R$.

We obtain under the limit $U\rightarrow \infty$ by the dominated convergence theorem
\begin{align}
&\phantom{\;=\;}\lim_{U\rightarrow \infty}\frac1{\epsilon_n^2U^3}\mathbb{E}\left[\int_{0}^{1}w_j(u)\mathcal{L}_{n,U}(u)e^{-iUux_j}\mathrm{d}u \overline{\int_{0}^{1}w_k(v)\mathcal{L}_{n,U}(v)e^{-iUvx_j}\mathrm{d}v}\right]\nonumber\\
&=\frac{2\pi\rho(\theta_j)^2}{T^2\exp(2T(\gamma-\lambda))}\int_{-\infty}^{\infty
} \left(-u^2w_j(u)\mathbf{1}_{[0,1]}(u)\right) \left(-u^2w_k(u)
 \mathbf{1}_{[0,1]}(u)\right) \mathrm{d}u\nonumber\\
&=\frac{2\pi\rho(\theta_j)^2}{T^2\exp(2T(\gamma-\lambda))}\int_{0}^{1}  u^4
w_j(u)w_k(u)\mathrm{d}u.\label{conjugated}
\end{align}
Without taking the complex conjugate in \eqref{convolution_I} we obtain
\begin{align}
&{T^2e^{2T(\gamma-\lambda)}}\mathbb{E}\left[\int_0^1w_j(u)\mathcal{L}_{n,U}(u)e^{-iUux_j}\mathrm{d}u \int_0^1w_k(v)\mathcal{L}_{n,U}(v)e^{-iUvx_j}\mathrm{d}v\right]\nonumber\\
&=2\pi\int_{-\infty}^{\infty} f_U(u)
\overline{(\overline{g_U(-v)} *
\mathcal{F}^{-1}(\rho(y/U+\theta_j)^2)(v))(u)}\mathrm{d}u.\nonumber
\end{align}
The same argumentation as before leads to
\begin{align}
&\phantom{\;=\;}\lim_{U\rightarrow \infty}\mathbb{E}\left[\int_{0}^{1}w_j(u)\mathcal{L}_{n,U}(u)e^{-iUux_j}\mathrm{d}u \int_{0}^{1}w_k(v)\mathcal{L}_{n,U}(v)e^{-iUvx_j}\mathrm{d}v\right]\nonumber\\
&=\frac{2\pi\rho(\theta_j)^2}{T^2\exp(2T(\gamma-\lambda))}\int_{-\infty}^{\infty
} \left(-u^2w_j(u)\mathbf{1}_{[0,1]}(u)\right) \left(-u^2w_k(-u)
\mathbf{1}_{[0,1]}(-u)\right) \mathrm{d}u\nonumber\\
&=0.\label{not_conjugated}
\end{align}
We combine \eqref{conjugated} and \eqref{not_conjugated} to obtain
\begin{align*}
&\phantom{\;=\;}\lim_{U\rightarrow \infty} \frac1{\epsilon_n^2U^3} \mathbb{E}\left[\mathrm{Re}\int_{0}^{1}\frac{w_j(u)\mathcal{L}_{n,U}(u)}{\exp(iUux_j)}\mathrm{d}u \:\mathrm{Re}\int_{0}^{1}\frac{w_k(u)\mathcal{L}_{n,U}(u)}{\exp(iUux_j)}\mathrm{d}v\right]\\
&=\lim_{U\rightarrow \infty}\frac1{\epsilon_n^2U^3} \mathbb{E}\left[\mathrm{Im}\int_{0}^{1}\frac{w_j(u)\mathcal{L}_{n,U}(u)}{\exp(iUux_j)}\mathrm{d}u \:\mathrm{Im}\int_{0}^{1}\frac{w_k(u)\mathcal{L}_{n,U}(u)}{\exp(iUux_j)}\mathrm{d}v\right]\\
&=\frac{\pi\rho(\theta_j)^2}{T^2\exp(2T(\gamma-\lambda))}\int_{0}^{1}  u^4
w_j(u)w_k(u)\mathrm{d}u.
\end{align*}
From \eqref{conjugated} and \eqref{not_conjugated} it also follows that the covariance between real and imaginary part vanishes asymptotically.

In the case $x_j\ne x_k$ we have to show that the covariance vanishes asymptotically. Without loss of generality we assume $x_j<x_k$. We define $\theta:=(\theta_j+\theta_k)/2$. Then $\theta_j<\theta<\theta_k$.
\begin{align*}
&\frac{T^2e^{2T(\gamma-\lambda)}}{\epsilon_n^2U^3} \mathbb{E}\left[\int_0^1w_j(u)\mathcal{L}_{n,U}(u)e^{-iUux_j}\mathrm{d}u \overline{\int_0^1w_k(v)\mathcal{L}_{n,U}(v)e^{-iUvx_k}\mathrm{d}v}\right]\\
&= \int_{-\infty}^\infty
\mathcal{F}f_U(U(x-\theta_j))\overline{\mathcal{F}g_U(U(x-\theta_k))}
\rho(x)^2U\mathrm{d}x\\
&= \int_{-\infty}^\infty
\mathcal{F}f_U(y+U(\theta-\theta_j))\overline{\mathcal{F}
g_U(y+U(\theta-\theta_k))} \rho(y/U+\theta)^2\mathrm{d}y
\end{align*}
By the Plancherel identity and by the dominated convergence theorem
\begin{align*}
\mathcal{F}f_U
\rightarrow \mathcal{F}\left(-u^2 w_j(u)\right)\quad\text{and}\quad\mathcal{F}g_U
\rightarrow \mathcal{F}\left(-u^2 w_k(u)\right)
\end{align*}
in $L^2(\R)$ for $U\to\infty$ and especially the $L^2(\R)$ norms converge. From
the assumption $\mathcal{F}\rho^2\in L^1(\R)$ follows that $\rho^2\in
L^\infty(\R)$. By the Cauchy--Schwarz inequality
\begin{align*}
&\lim_{U\to\infty}\left|\int_{-\infty}^0
\mathcal{F}f_U(y+U(\theta-\theta_j))\overline{\mathcal{F}
g_U(y+U(\theta-\theta_k))} \rho(y/U+\theta)^2\mathrm{d}y\right|\\
&\le\lim_{U\to\infty}\|\rho\|_\infty^2\left\|\mathcal{F} f_U\right\|_{L^2(\R)}
\left(\int_{-\infty}^{U(\theta-\theta_k)}\left|\mathcal{F} g_U(y)\right|^2\mathrm{d}y\right)^{1/2}=0.
\end{align*}
A similar calculation shows that the integral over $(0,\infty)$ converges to zero and consequently
\begin{align*}
&\lim_{U\to\infty}\frac{T^2e^{2T(\gamma-\lambda)}}{\epsilon_n^2U^3} \mathbb{E}\left[\int_0^1\frac{w_j(u)\mathcal{L}_{n,U}(u)}{\exp(iUux_j)}\mathrm{d}u \overline{\int_0^1\frac{w_k(v)\mathcal{L}_{n,U}(v)}{\exp(iUvx_k)}}\mathrm{d}v\right]=0.
\end{align*}
The same way follows
\begin{align*}
&\lim_{U\to\infty}\frac{T^2e^{2T(\gamma-\lambda)}}{\epsilon_n^2U^3} \mathbb{E}\left[\int_0^1\frac{w_j(u)\mathcal{L}_{n,U}(u)}{\exp(iUux_j)}\mathrm{d}u \int_0^1\frac{w_k(v)\mathcal{L}_{n,U}(v)}{\exp(iUvx_k)}\mathrm{d}v\right]=0.
\end{align*}

The $1/(\epsilon_n U^{3/2}) \int_{0}^{1}w_j(u)\mathcal{L}_{n,U}(u)e^{-iUux_j}\mathrm{d}u$ are centered normal random variables and their covariance matrix converges to the covariance matrix of the claimed limit. Thus, the characteristic function converges pointwise.
By L\'evy's continuity theorem this shows the convergence in distribution.
\qed
\end{proof}

\begin{lemma}\label{covariances}
Let $\sigma>0$ and $\rho \in L^\infty(\R)$.
Let $w_U, \tilde w_U\in L^\infty ([0,1],\C)$ be Riemann--integrable and let
there be a constant $C>0$ such that $\|w_U\|_\infty, \|\tilde w_U\|_\infty\le C$
for all $U\ge1$.
Let there be $a,\tilde a: [1,\infty)\rightarrow \C$ such that the condition
\begin{equation}\label{condition}
\lim_{\delta\to 0}\sup_{U\ge 1}\sup_{u\in[1-\delta/U,1]}|w_U(u)-a(U)|=0
\end{equation}
and the corresponding condition for $\tilde w_U$ and $\tilde a$ hold.
Then
\begin{align}
\lim_{U\to\infty} & \frac{1}{\epsilon_n^2 \exp(T\sigma^2U^2)} \mathbb{E}\left[\int_0^1w_U(u)\mathcal{L}_{n,U}(u)\mathrm{d}u \int_0^1\tilde w_U(v)\mathcal{L}_{n,U}(v)\mathrm{d}v\right]=0,\nonumber\\
\lim_{U\to\infty} \biggl(&\frac{1}{\epsilon_n^2 \exp(T\sigma^2U^2)} \mathbb{E}\left[\int_0^1w_U(u)\mathcal{L}_{n,U}(u)\mathrm{d}u \overline{\int_0^1 \tilde w_U(v)\mathcal{L}_{n,U}(v)\mathrm{d}v}\right]\nonumber\\
&-\frac{a(U)\overline{\tilde a(U)}\int_{-\infty}^\infty
\rho(y)^2\mathrm{d}y}{\exp(2T(\sigma^2/2+\gamma-\lambda)) T^4
\sigma^4}\biggr)=0.\nonumber
\end{align}
\end{lemma}

\begin{remark}
Obviously $a(U):=w_U(1)$ is the only possible definition. Thus, $a$ describes the dependence of $w_U$ on $U$ at one.
\end{remark}

\begin{proof}
We notice that \eqref{ito} applies to the complex--valued functions and yields for $w_j:=w_U$ and $w_k:=\tilde w_U$ with the definitions \eqref{f} and \eqref{g} of $f_U$ and $g_U$, respectively, and with $\theta:=T\sigma^2+T\gamma$
\begin{align}
&\quad{T^2e^{2T(\sigma^2/2+\gamma-\lambda)}} \mathbb{E}\left[\int_0^1w_U(u)\mathcal{L}_{n,U}(u)\mathrm{d}u \overline{\int_0^1\tilde w_U(v)\mathcal{L}_{n,U}(v)\mathrm{d}v}\right]\nonumber\\
=&{\epsilon_n^2U^4} \int_{-\infty}^\infty \mathcal{F}(f_U(u)e^{T\sigma^2U^2u^2/2})(U(x-\theta))
\overline{\mathcal{F}(g_U(v)e^{T\sigma^2U^2v^2/2})(U(x-\theta))} \rho(x)^2
\mathrm{d}x\nonumber\\
=&{2\pi\epsilon_n^2U^3} \int_{-\infty}^\infty f_U(u)e^{T\sigma^2U^2u^2/2}
\left(\overline{g_U(v)}e^{T\sigma^2U^2v^2/2}*\overline{\mathcal{F}^{-1}(
\rho(y/U+\theta)^2)(v)}\right)(u) \mathrm{d}u\nonumber\\
=&{2\pi\epsilon_n^2U^4} \int_{0}^1 \int_0^1 f_U(u)\overline{g_U(v)}
\overline{\mathcal{F}^{-1}( \rho(y+\theta)^2)(U(u-v))}
e^{T\sigma^2U^2(u^2+v^2)/2}\mathrm{d}u\mathrm{d}v.\label{covariance}
\end{align}
For all $\delta>0$ we have
\begin{align}
&\phantom{\;=\;}
\lim_{U\to\infty}T\sigma^2U^2e^{-T\sigma^2U^2/2}\int_{1-\delta/U}^1u e^{T\sigma^2U^2u^2/2}\mathrm{d}u\nonumber\\
&
=\lim_{U\to\infty}e^{-T\sigma^2U^2/2}\left[ e^{T\sigma^2U^2u^2/2}\right]_{1-\delta/U}^1
=1-\lim_{U\to\infty}e^{-T\sigma^2U \delta+T\sigma^2\delta^2/2}=1.\label{dirac}
\end{align}
For the product of two such sequences we obtain for all $\delta>0$
\begin{align}
\lim_{U\to\infty}T^2\sigma^4U^4e^{-T\sigma^2U^2}\int_{1-\delta/U}^1\int_{1-\delta/U}^1uv e^{T\sigma^2U^2(u^2+v^2)/2}\mathrm{d}u\mathrm{d}v=1.\label{one}
\end{align}
Likewise
\begin{align}
\lim_{U\to\infty}T^2\sigma^4U^4e^{-T\sigma^2U^2}\int_{0}^1\int_{0}^1uv e^{T\sigma^2U^2(u^2+v^2)/2}\mathrm{d}u\mathrm{d}v=1\label{two}
\end{align}
holds. We scale the integral in \eqref{covariance} appropriately:
\begin{align}
&\phantom{\;=\;}\lim_{U\to\infty}\biggl({T^2\sigma^4U^4e^{-T\sigma^2U^2}} \int_{0}^1 \int_0^1 f_U(u)\overline{g_U(v)}\label{limit_I}\\
&\phantom{\;=\;}
\overline{\mathcal{F}^{-1}( \rho(y+\theta)^2)(U(u-v))}
e^{T\sigma^2U^2(u^2+v^2)/2}\mathrm{d}u\mathrm{d}v
-a(U)\overline{\tilde a(U)} \frac1{2\pi}
\int_{-\infty}^{\infty}\rho(y)^2\mathrm{d}y\biggr)
\ifodd\switch
\nonumber\\
\else
\nonumber\displaybreak[0]\\
\fi
&=\lim_{U\to\infty}\biggl({T^2\sigma^4U^4e^{-T\sigma^2U^2}} \int_{0}^1 \int_0^1 uve^{T\sigma^2U^2(u^2+v^2)/2}\label{limit_II}\\
&\phantom{\;=\;}\overline{\mathcal{F}^{-1}( \rho(y+\theta)^2)(U(u-v))}
f_U(u)\overline{g_U(v)}/(uv)\mathrm{d}u\mathrm{d}v
-a(U)\overline{\tilde a(U)} \frac1{2\pi}
\int_{-\infty}^{\infty}\rho(y)^2\mathrm{d}y\biggr)\nonumber\\
\intertext{We recall that in the Gaussian white noise model we assumed $\rho$ to
be in $L^2(\R)$. Since
$\overline{\mathcal{F}^{-1}(\rho(y+\theta)^2)(U(u-v))}f_U(u)\overline{g_U(v)}
/(uv)$ is bounded in $L^\infty([0,1]^2)$ independently of $U$ for $U\ge1$ and
since the difference between \eqref{two} and \eqref{one} is zero, only the
integral over $[1-\delta/U,1]^2$ contributes to the limit. For all $\delta>0$
the limit equals}
&=\lim_{U\to\infty}\biggl(T^2\sigma^4U^4e^{-T\sigma^2U^2}\int_{1-\delta/U}^{1} \int_{1-\delta/U}^1 uve^{T\sigma^2U^2(u^2+v^2)/2}\label{dirac limit}\\
&\phantom{\;=\;} \overline{\mathcal{F}^{-1}(\rho(y+\theta)^2)(U(u-v))}
f_U(u)\overline{g_U(v)}/(uv)\mathrm{d}v\mathrm{d}u
-a(U)\overline{\tilde a(U)} \frac1{2\pi}
\int_{-\infty}^{\infty}\rho(y)^2\mathrm{d}y\biggr)\nonumber\\
&=0,\nonumber
\end{align}
which can be seen the following way.
$\mathcal{F}^{-1}(\rho(y+\theta)^2)$ is continuous and we
have $|U(u-v)|\le\delta$ for
all $u,v\in[1-\delta/U,1]$. 
$\mathcal{F}^{-1}(\rho(y+\theta)^2)(U(u-v))$ gets arbitrarily close to
$\mathcal{F}^{-1}(\rho(y+\theta)^2)(0)=(1/2\pi)\int_{-\infty}
^\infty\rho(y)^2\mathrm{d}y$ by choosing $\delta$ small enough. By
\eqref{condition}, $w_U(u)$ tends to $a(U)$
and $\tilde w_U(v)$ tends to $\tilde a(U)$ for $\delta$ tending to zero. By
choosing $U$ large the factor $(-u+i/U)/\exp(T\mathcal{F}\mu(Uu))$ gets close to
minus one for all $u\in[1-\delta/U,1]$.
Thus, for small $\delta$ and large $U$ the term $f_U(u)\overline{g_U(v)}/(uv)$ is close to $a(U)\overline{\tilde a(U)}$ for all $u,v\in[1-\delta/U,1]$.

Rescaling \eqref{covariance} and taking the limit $U\to\infty$ leads to
\begin{align}
&\phantom{\;=\;}\lim_{U\to\infty} \biggl(\frac{1}{\epsilon_n^2 \exp(T\sigma^2U^2)} \mathbb{E}\left[\int_0^1w_U(u)\mathcal{L}_{n,U}(u)\mathrm{d}u \overline{\int_0^1 \tilde w_U(v)\mathcal{L}_{n,U}(v)\mathrm{d}v}\right]\nonumber\\
&\phantom{\;=\;}-\frac{a(U)\overline{\tilde a(U)}\int_{-\infty}^\infty
\rho(y)^2\mathrm{d}y}{\exp(2T(\sigma^2/2+\gamma-\lambda)) T^4
\sigma^4}\biggr)\nonumber\\
&=\lim_{U\to\infty}\biggl(\frac{2\pi U^4\exp(-T\sigma^2U^2)}{T^2\exp(2T(\sigma^2/2+\gamma-\lambda))} \int_{0}^1 \int_0^1 f_U(u)\overline{g_U(v)}\nonumber\\
&\phantom{\;=\;}\overline{\mathcal{F}^{-1}( \rho(y+\theta_0)^2)(U(u-v))}
e^{T\sigma^2U^2(u^2+v^2)/2}\mathrm{d}u\mathrm{d}v\nonumber\\
&\phantom{\;=\;}-\frac{2\pi}
{\exp(2T(\sigma^2/2+\gamma-\lambda))}\frac{a(U)\overline{\tilde
a(U)}}{T^4\sigma^4 2\pi} \int_{-\infty}^\infty
\rho(y)^2\mathrm{d}y\biggr)=0,\label{with_conjugation}
\end{align}
where we used that \eqref{limit_I} is zero.
By \eqref{ito_II} with $y=U(x-\theta)$ we have
\begin{align*}
&\quad{T^2e^{2T(\sigma^2/2+\gamma-\lambda)}} \mathbb{E}\left[\int_0^1w_U(u)\mathcal{L}_{n,U}(u)\mathrm{d}u \int_0^1\tilde w_U(v)\mathcal{L}_{n,U}(v)\mathrm{d}v\right]\nonumber\\
=&{\epsilon_n^2U^3} \int_{-\infty}^\infty \mathcal{F}(f_U(u)e^{T\sigma^2U^2u^2/2})(y)
\overline{\mathcal{F}(\overline{g_U(-v)}e^{T\sigma^2U^2v^2/2})(y) }
\rho(y/U+\theta)^2 \mathrm{d}y\nonumber\displaybreak[0]\\
=&{2\pi\epsilon_n^2U^3} \int_{-\infty}^\infty f_U(u)e^{T\sigma^2U^2u^2/2}
\left(g_U(-v)e^{T\sigma^2U^2v^2/2}*\overline{\mathcal{F}^{-1}(
\rho(y/U+\theta)^2)(v)}\right)(u) \mathrm{d}u\displaybreak[0]\nonumber\\
=&{2\pi\epsilon_n^2U^4} \int_{0}^1 \int_{-1}^0 f_U(u)g_U(-v)
\overline{\mathcal{F}^{-1}( \rho(y+\theta)^2)(U(u-v))}
e^{T\sigma^2U^2(u^2+v^2)/2}\mathrm{d}v\mathrm{d}u\nonumber\\
=&{2\pi\epsilon_n^2U^4} \int_{0}^1 \int_{0}^1 uv e^{T\sigma^2U^2(u^2+v^2)/2}
\mathcal{F}^{-1}( \rho(y+\theta)^2)(-U(u+v))\frac{f_U(u)g_U(v)}{uv}
\mathrm{d}v\mathrm{d}u.
\end{align*}
Rescaling as in \eqref{with_conjugation} leads to
\begin{align}
&\phantom{\;=\;}\lim_{U\to\infty} \frac{1}{\epsilon_n^2 \exp(T\sigma^2U^2)} \mathbb{E}\left[\int_0^1w_U(u)\mathcal{L}_{n,U}(u)\mathrm{d}u \int_0^1\tilde w_U(v)\mathcal{L}_{n,U}(v)\mathrm{d}v\right]\nonumber\\
&=\lim_{U\to\infty}\frac{2\pi U^4\exp(-T\sigma^2U^2)}{T^2\exp(2T(\sigma^2/2+\gamma-\lambda))} \int_{0}^1 \int_0^1 uv e^{T\sigma^2U^2(u^2+v^2)/2}\nonumber\\
&\phantom{\;=\;}\mathcal{F}^{-1}( \rho(y+\theta)^2)(-U(u+v)) f_U(u) g_U(v)/(uv)
\mathrm{d}u\mathrm{d}v=0,\label{without_conjugation}
\end{align}
since $\mathcal{F}^{-1}(\rho(y+\theta)^2)(u)\to 0$ for $|u|\to\infty$.
\qed
\end{proof}

\begin{lemma}\label{linearized_>0}
Let $\sigma>0$, $\rho\in L^\infty(\R)$ and
$x_0\in\R$. For $j=1,\dots,n $ let $w_j:[0,1]\to\R$ be continuous at
one, Riemann--integrable and in $L^\infty([0,1])$.
Then
\[\frac1{\epsilon_n \exp(T\sigma^2U^2/2)}\int_{0}^{1}w_j(u)\mathcal{L}_{n,U}(u)e^{-iUux_0}\mathrm{d}u\] converge jointly in distribution to
\[\frac{\|\rho\|_{L^2(\R)} w_j(1)}
{\sqrt{2}\exp(T(\sigma^2/2+\gamma-\lambda))T^2\sigma^2}\left(W+iV\right) \]
as $U\rightarrow \infty$, where $W$ and $V$ are independent standard normal random variables.
\end{lemma}

\begin{proof}
\ifodd\switch
The proof relies on Lemma~\ref{covariances}.
We define $w_U(u):=w_j(u)/\exp(iUux_0)$ and $\tilde
w_U(u):=w_k(u)/\exp(iUux_0)$. Further we set $a(U):=w_j(1)/\exp(iUx_0)$ and
define $\tilde a(U):=w_k(1)/\exp(iUx_0)$. Then we apply Lemma~\ref{covariances}.
Condition~\eqref{condition} is satisfied since $w_j$ and $w_k$ are continuous at
one and since
\else
The proof relies on Lemma~\ref{covariances}.
We define $w_U(u):=w_j(u)/\exp(iUux_0)$ and $\tilde
w_U(u):=w_k(u)/\exp(iUux_0)$. Further we set $a(U):=w_j(1)/\exp(iUx_0)$ and
$\tilde a(U):=w_k(1)/\exp(iUx_0)$. Then we apply Lemma~\ref{covariances}.
Condition~\eqref{condition} is satisfied since $w_j$ and $w_k$ are continuous at
one and since
\fi
\[\exp(-iUux_0)=\exp(-iUx_0)\exp(iU(1-u)x_0),
\]
where $U(1-u)\le\delta$ for
$u\in[1-\delta/U,1]$. We note that $a(U)\overline{\tilde a(U)}=w_j(1)w_k(1)$ is
real. By Lemma~\ref{covariances} the covariances converge to the covariances of
the claimed limit. The convergence in distribution follows by L\'evy's
continuity theorem.
\qed
\end{proof}

\begin{lemma}\label{linearized_no_convergence}
Let $\sigma>0$ and $\rho\in L^\infty(\R)$.
Take $w_1,w_2:[0,1]\to\R$ to be Riemann--integrable, in $L^\infty([0,1])$ and
continuous at one.
Let $x_1,x_2\in\R$ and denote $x_2-x_1$ by $\varphi$.
Then
\begin{align*}
  \frac{1}{\epsilon_n e^{T\sigma^2U^2/2}}\biggl(w_1(1)\int_0^1 \frac{w_2(u) \mathcal{L}_{n,U}(u)}{e^{iUux_2}}\mathrm{d}u
  -\frac{w_2(1)}{e^{iU\varphi}}\int_0^1 \frac{w_1(u)
\mathcal{L}_{n,U}(u)}{e^{iUux_1}}\mathrm{d}u\biggr)\xrightarrow{\mathbb{P}} 0
\end{align*}
as $U\to\infty$.
\end{lemma}

\begin{proof}
This lemma is a consequence of Lemma~\ref{covariances}.
We define
\[w_U(u):=\frac{w_1(1)w_2(u)}{\exp(iUux_2)}-\frac{w_2(1)w_1(u)}{\exp(iU\varphi)\exp(iUux_1)}.\]
$w_U(u)$ fulfills condition~\eqref{condition} with $a(U)=0$ for all $U\ge1$. Lemma~\ref{covariances} yields
\[
\lim_{U\to\infty} \frac{1}{\epsilon_n^2 \exp(T\sigma^2U^2)} \mathbb{E}\left[\left|\int_0^1w_U(u)\mathcal{L}_{n,U}(u)\mathrm{d}u \right|^2\right]=0
\]
and the statement follows by L\'evy's continuity theorem.
\qed
\end{proof}

\subsection{The remainder term}

In this section, we show that the contribution of the remainder term to the estimation vanishes asymptotically. We recall that the remainder term $\mathcal{R}_{n,U}$ depends on the L\'evy triplet.

\begin{lemma}\label{linearization_>0}
Let $\sigma_0>0$.
Let $w_U\in L^\infty ([0,1],\C)$ be Riemann--integrable and let there be a constant $C>0$ such that $\|w_U\|_\infty\le C$ for all $U\ge1$.
If $\epsilon_n U_n^2\exp(T\sigma_0^2U_n^2/2) \to 0$ as $n\to \infty$, then for all L\'evy triplets with $\sigma\le\sigma_0$
\begin{align*}
\frac1{\epsilon_n\exp(T\sigma_0^2U_n^2/2)}&\int_{0}^1w_{U_n}(u) \mathcal{R}_{n,U_n}(u) \mathrm{d}u \xrightarrow{\mathbb{P}}0, \qquad \text{as } n\to\infty.
\end{align*}

\end{lemma}

\begin{proof}
By the identity $\mathcal{R}_{n,U}(u) =(1/T)\log_{\ge\kappa^{U}(u)}(1+T\mathcal{L}_{n,U}(u))-\mathcal{L}_{n,U}(u)$, where $\kappa^{U}(u)\le1/2$, we have to show that for $U=U_n$
\begin{equation}\label{remainder}
\frac1{\epsilon_n\exp(T\sigma_0^2U^2/2)}\int_0^1 w_U(u)\left(\log_{\ge\kappa^{U}(u)}(1+T\mathcal{L}_{n,U}(u))-T\mathcal{L}_{n,U}(u)\right) \mathrm{d}u
\end{equation}
converges in probability to zero.
For $z\in\C$ holds $\log(1+z)-z=O(|z|^2)$ as $|z|\to0$. We define $g$ by
$g(z):=(\log(1+z)-z)/|z|^2$ for $z\neq0$ and $g(0):=0$. There are $M$ and
$\eta>0$ such that $|g(z)|\le M$ for all $|z|\le\eta$. We may assume that
$\eta\le1/2$.
If the logarithm in the definition of $g$ is replaced by the trimmed logarithm
$\log_{\ge\kappa}$ with some $\kappa\in(0,1/2]$ then $g$ remains unchanged for
$|z|\le1/2$. Thus, the statement holds uniformly for all
$g_\kappa(z):=(\log_{\ge\kappa}(1+z)-z)/|z|^2$ with $\kappa\in(0,1/2]$.

By Proposition~\ref{sup L} we have $\sup_{u\in[-1,1]}|\mathcal{L}_{n,U}(u)| \xrightarrow{\mathbb{P}}0$. Let $\tau>0$ be given. Eventually we have
\begin{align*}
&\phantom{\le}\mathbb{P}\left(\exists u\in[-1,1]: {|\log_{\ge\kappa^{U}(u)}(1+T\mathcal{L}_{n,U}(u))-T\mathcal{L}_{n,U}(u)|} >M T^2 {|\mathcal{L}_{n,U}(u)|^2}\right)\\
&\le\mathbb{P}\left(T\sup_{u\in[-1,1]}|\mathcal{L}_{n,U}(u)|>\eta\right)<\tau.
\end{align*}
Except on a set with probability less than $\tau$ we have eventually
\begin{align}
&\phantom{le}\frac1{\epsilon_n\exp(T\sigma_0^2U^2/2)}\left|\int_0^1 w_U(u) \left(\log_{\ge\kappa^{U}(u)}(1+T\mathcal{L}_{n,U}(u))-T\mathcal{L}_{n,U}(u)\right) \mathrm{d}u\right|\nonumber\\
&\le  \frac{MT^2}{\epsilon_n\exp(T\sigma_0^2U^2/2)} \int_0^1 |w_U(u)\mathcal{L}_{n,U}(u)^2|\mathrm{d}u.\label{dominator}
\end{align}
Hence \eqref{remainder} converges in probability to zero if \eqref{dominator} converges in probability to zero.
The convergence
\[\frac1{\epsilon_n\exp(T\sigma_0^2U^2/2)}\int_{0}^1 |w_U(u)\mathcal{L}_{n,U}(u)^2|\mathrm{d}u\to0\] holds even in $L^1$ since
\begin{align}
&\phantom{\;=\;}\frac1{\epsilon_n\exp(T\sigma_0^2U^2/2)}\mathbb{E}\left[\int_{0}^1 |w_U(u)\mathcal{L}_{n,U}(u)^2|\mathrm{d}u\right]\label{dominator II}\\
&\le\frac{C}{\epsilon_n\exp(T\sigma_0^2U^2/2)}\mathbb{E}\left[\int_{0}^1 |\mathcal{L}_{n,U}(u)^2|\mathrm{d}u\right]
\ifodd\switch
\nonumber\\
\else
\nonumber\displaybreak[0]\\
\fi
&\le\frac{C}{\epsilon_n\exp(T\sigma_0^2U^2/2)}\nonumber\\
&\phantom{\;=\;}\int_{0}^1\left|\frac{\epsilon_n \: iUu(1+iUu)}{T(1+iUu(1+iUu)\mathcal{FO}(Uu))}\right|^2
\mathbb{E}\left[\left|\int_{-\infty}^\infty
e^{iUux}\rho(x)\mathrm{d}W(x)\right|^2\right]\mathrm{d}u\nonumber\\
&\le \frac{C}{\epsilon_n\exp(T\sigma_0^2U^2/2)}\int_0^1 \frac{\epsilon_n^2
(U^2+U^4)u\exp(T\sigma^2U^2u^2)\|\rho\|_{L^2(\R)}^2}
{T^2\exp(2T(\sigma^2/2+\gamma-\lambda)-2T\|\mathcal{F}\mu\|_\infty)}\mathrm{d}u,
\label{dominator III}\\
\intertext{for $\sigma=0$ this converges to zero and for $\sigma>0$ we further calculate, }
&= \frac{C\epsilon_n (1+U^2)\|\rho\|_{L^2(\R)}^2
\int_0^1 2T\sigma^2U^2u\exp(T\sigma^2U^2u^2)\mathrm{d}u} {\exp(T\sigma_0^2U^2/2)2T^3\sigma^2 \exp(2T(\sigma^2/2+\gamma-\lambda)-2T\|\mathcal{F}\mu\|_\infty)} \nonumber\\
&= \frac{C\epsilon_n (1+U^2)\|\rho\|_{L^2(\R)}^2
(\exp(T\sigma^2U^2)-1)} {\exp(T\sigma_0^2U^2/2)2T^3\sigma^2 \exp(2T(\sigma^2/2+\gamma-\lambda)-2T\|\mathcal{F}\mu\|_\infty)} \nonumber\\
&\le \frac{C\|\rho\|_{L^2(\R)}^2\epsilon_n (1+U^2)
(\exp(T\sigma_0^2U^2/2))} {2T^3\sigma^2 \exp(2T(\sigma^2/2+\gamma-\lambda)-2T\|\mathcal{F}\mu\|_\infty)}  \to 0\nonumber
\end{align}
as $n\to\infty$. Thus, \eqref{remainder} converges in probability to zero.
\qed
\end{proof}

\begin{lemma}\label{linearization_0}
Let $w_U\in L^\infty ([0,1],\C)$ be Riemann--integrable and let there be a constant $C>0$ such that $\|w_U\|_\infty\le C$ for all $U\ge1$.
If $U_n\to\infty$ and $\epsilon_n U_n^{5/2} \to 0$ as $n\to\infty$, then for all L\'evy triplets with $\sigma=0$
\begin{align*}
\frac1{\epsilon_n U_n^{3/2}}&\int_{0}^1w_{U_n}(u)
\mathcal{R}_{\epsilon_n,U_n}(u) \mathrm{d}u  \xrightarrow{\mathbb{P}}0
\end{align*}
as $n\to\infty$.
\end{lemma}
\begin{proof}
We follow the proof of Lemma~\ref{linearization_>0}. $\sup_{u\in[-1,1]}|\mathcal{L}_{n,U}(u)| \xrightarrow{\mathbb{P}}0$ holds by Proposition~\ref{sup L}. We set $\sigma_0=0$ and divide by $U^{3/2}$ in \eqref{remainder} and \eqref{dominator}.
Then we use that \eqref{dominator II} is bounded by \eqref{dominator III}, where we set $\sigma_0=\sigma=0$ and divide by $U^{3/2}$ again.
We obtain
\[
\frac1{\epsilon_n U^{3/2}}\mathbb{E}\left[\int_{0}^1 |w_U(u)\mathcal{L}_{n,U}(u)^2|\mathrm{d}u\right]
\le  \frac{C\epsilon_n\:(U^{1/2}+U^{5/2}) \|\rho\|_{L^2(\R)}^2}
{T^2\exp(T(2(\gamma-\lambda)-2\|\mathcal{F}\mu\|_\infty ))} \rightarrow  0
\]
as $\epsilon_n\rightarrow0$, which implies the desired convergence.
\qed
\end{proof}

\subsection{The approximation errors}

The approximation error can be controlled as in \cite{calibration} using the order conditions \eqref{order s} on the weight functions.
The L\'evy triplet $\mathcal{T}=(\sigma^2,\gamma,\mu)$ was assumed to be contained in $\mathcal{G}_s(R,\sigma_{\max})$, especially $\mu$ is s--times weakly differentiable and we have $\max_{0\le k\le s}\|\mu^{(k)}\|_{L^2(\mathbb{R})}\le R$, $\|\mu^{(s)}\|_{\infty}\le R$.

We use $(iu)^s\mathcal{F}\mu(u)=\mathcal{F}\mu^{(s)}(u)$ and the Plancherel identity to bound the approximation error by
\ifodd\switch
\begin{align}
&\phantom{\;=\;}
\left|\frac2{U^2}\int_{0}^{1}\mathrm{Re}(\mathcal{F}\mu(Uu)) w_{\sigma}^{1}(u)\mathrm{d}u\right|
=\frac{1}{U^2}\left|\int_{-1}^{1}\mathcal{F}\mu(Uu) w_{\sigma}^{1}(u)\mathrm{d}u\right|\nonumber\\
&=\frac{2\pi}{U^2}\left|\int_{-\infty}^\infty\mu^{(s)}(x/U)U^{-1} \overline{\mathcal{F}^{-1}(w_{\sigma}^{1}(u)/(iUu)^s)(x)}\mathrm{d}x\right|\nonumber\\
&
\le U^{-(s+3)} \|\mu^{(s)}\|_{\infty}\|\mathcal{F}(w^1_\sigma(u)/u^s)\|_{L^1(\R)}.\label{bias_sigma}
\end{align}
\else
\begin{align}
&\phantom{\;=\;}
\left|\frac2{U^2}\int_{0}^{1}\mathrm{Re}(\mathcal{F}\mu(Uu)) w_{\sigma}^{1}(u)\mathrm{d}u\right|
=\frac{1}{U^2}\left|\int_{-1}^{1}\mathcal{F}\mu(Uu) w_{\sigma}^{1}(u)\mathrm{d}u\right|\nonumber\displaybreak[0]\\
&=\frac{2\pi}{U^2}\left|\int_{-\infty}^\infty\mu^{(s)}(x/U)U^{-1} \overline{\mathcal{F}^{-1}(w_{\sigma}^{1}(u)/(iUu)^s)(x)}\mathrm{d}x\right|\nonumber\displaybreak[0]\\
&
\le U^{-(s+3)} \|\mu^{(s)}\|_{\infty}\|\mathcal{F}(w^1_\sigma(u)/u^s)\|_{L^1(\R)}.\label{bias_sigma}
\end{align}
\fi
Analogously we obtain
\begin{align}
\phantom{\;=\;}\left|\frac2{U} \int_{0}^{1}\mathrm{Im}(\mathcal{F}\mu(Uu))w_{\gamma}^{1}(u)\mathrm{d}u\right|
&\le U^{-(s+2)} \|\mu^{(s)}\|_{\infty}\|\mathcal{F}(w^1_\gamma(u)/u^s)\|_{L^1(\R)},\label{bias_gamma}\\
\left|2\int_{0}^{1}\mathrm{Re}(\mathcal{F}\mu(Uu))w_{\lambda}^{1}(u)\mathrm{d}u\right|
&\le U^{-(s+1)} \|\mu^{(s)}\|_{\infty}\|\mathcal{F}(w^1_\lambda(u)/u^s)\|_{L^1(\R)}.\label{bias_lambda}
\end{align}
The last error term in \eqref{error_mu} can be bounded by
\begin{align}
&\left|U\mathcal{F}^{-1}\left[(1-w_\mu^1(u)) \mathcal{F}\mu(Uu)\right](Ux)\right|
=\frac{U}{2\pi}\left|\int_{-\infty}^\infty (1-w_\mu^1(u))\mathcal{F}\mu(Uu)e^{-iUux}\mathrm{d}u\right|\nonumber\\
&
=\frac1{2\pi U^{s-1}}\left|\int_{-\infty}^\infty \overline{\frac{1-w_\mu^1(u)}{u^{s}}e^{iUux}}\mathcal{F}\mu^{(s)}(Uu)\mathrm{d}u\right|\nonumber\\
&= U^{-s}\left|\int_{-\infty}^\infty \overline{\mathcal{F}^{-1}\left(\frac{1-w_\mu^1(u)}{u^{s}}e^{iUux}\right)(y)} \mu^{(s)}\left(\frac{y}{U}\right)\mathrm{d}y\right|\nonumber\\
&
\le \frac{\|\mu^{(s)}\|_{\infty}}{2\pi U^{s}} \left\| \mathcal{F}\left(\frac{1-w_\mu^1(u)}{u^{s}}\right) \right\|_{L^1(\R)}. \label{bias_mu}
\end{align}